\documentclass[journal,twoside,web]{ieeecolor}
\usepackage{generic}
\usepackage{cite}
\usepackage{textcomp}

\usepackage{amsmath,amssymb,amsfonts}
\usepackage{algorithmic}
\usepackage{graphicx}
\usepackage{textcomp}
\usepackage{subcaption} 
\usepackage[T1]{fontenc}
\usepackage[utf8]{inputenc}
\usepackage{pgfplots}
\usepackage{grffile}
\pgfplotsset{compat=newest}
\usetikzlibrary{plotmarks}
\usetikzlibrary{arrows.meta}
\usepgfplotslibrary{patchplots}
\usepackage{multirow}
\usepackage{booktabs}
\usepackage{tablefootnote}   
\usepackage{amsmath,amssymb}
\usepackage{amssymb,latexsym,amsfonts,amsmath,bbm}
\usepackage{mathrsfs}
\usepackage{paralist}
\usepackage{dsfont}
\usepackage{eso-pic}
\usepackage{epsfig}
\usepackage{mathtools}
\usepackage{epstopdf}
\usepackage{lipsum}

\usepackage[nocomma]{optidef}

\usepackage{tikz}
\usetikzlibrary{calc,shapes,arrows}
\usepackage{tikz}
\usepackage{bm}
\usepackage[colorlinks=true,allcolors=blue]{hyperref}

\usepackage{algorithm}
\newcommand{\EE}{\mathds{E}}
\newcommand{\PP}{\mathds{P}}

\newcounter{theorem}
\newcounter{definition}
\newcounter{lemma}
\newcounter{claim}
\newcounter{problem}
\newcounter{proposition}
\newcounter{corollary}
\newcounter{construction}
\newcounter{example}
\newcounter{xca}
\newcounter{comments}
\newcounter{remark}
\newcounter{assumption}

\newtheorem{theorem}[theorem]{Theorem}
\newtheorem{lemma}[lemma]{Lemma}

\newtheorem{definition}[definition]{Definition}
\newtheorem{example}[example]{Example}

\newtheorem{remark}[remark]{Remark}

\usepackage{xcolor}
\usepackage{cite}
\usepackage[colorlinks=true,
allcolors=blue]{hyperref}

\makeatletter
\renewcommand\@cite[2]{\textcolor{blue}{[{#1\if@tempswa , #2\fi}]}}
\makeatother

\usepackage{etoolbox}

\makeatletter
\AtBeginDocument{%
	\pagestyle{headings}
	
	\patchcmd{\@oddhead}{\\[-19pt]}{\\[-8pt]}{}{}%
	\patchcmd{\@evenhead}{\\[-19pt]}{\\[-8pt]}{}{}%
}
\makeatother

\begin{document}
	
	\title{Safe Packetized Control for Stochastic Constrained Networked Systems}
	
	\author{
		Omid Akbarzadeh, \IEEEmembership{Student Member,~IEEE},
		MohammadHossein Ashoori, \IEEEmembership{Student Member,~IEEE},
		Mohammad H. Mamduhi, \IEEEmembership{Senior Member,~IEEE},
		and Abolfazl Lavaei, \IEEEmembership{Senior Member,~IEEE}
		\thanks{
			O. Akbarzadeh, M.H. Ashoori, and A. Lavaei are with the School of Computing, Newcastle University, United Kingdom.
			Emails:
			\texttt{o.akbarzadeh2@newcastle.ac.uk},
			\texttt{m.ashoori2@newcastle.ac.uk},\\
			\texttt{abolfazl.lavaei@newcastle.ac.uk}.
		}
		\thanks{
			M.H. Mamduhi is with the School of Computer Science, University of Birmingham, United Kingdom.
			Email: \texttt{m.h.mamduhi@bham.ac.uk}.
		}
	}
	
	\maketitle

\begin{abstract}
This work develops a formal framework for the synthesis of \emph{packetized} safety controllers for discrete-time polynomial stochastic networked control systems (dt-PSNCS) operating under communication constraints, including {uplink delays} (plant-to-controller) and {downlink packet losses} (controller-to-actuator). In this setting, the controller is deployed remotely and exchanges information with the plant over an imperfect wireless communication network. Our proposed approach treats the downlink channel as an {erasure channel}, with packet losses characterized by an independent Bernoulli process. To systematically manage both uplink delays and downlink packet loss, we first introduce a \emph{buffer} collocated with the plant that accommodates the packetized safety control (PSC) mechanism. We augment the plant and buffer states into a unified augmented-state representation that accurately captures the system evolution in the presence of communication imperfections. Our proposed framework synthesizes safety controllers based on control barrier certificates (CBCs), providing probabilistic safety guarantees that remain robust in the presence of both communication delays and packet losses. To achieve this, we reformulate the safety constraints as a sum-of-squares (SOS) optimization program, thereby facilitating the systematic construction of CBCs and their corresponding safety controllers. We validate the proposed framework through three (physical) case studies, demonstrating its effectiveness and practical applicability.
\end{abstract}


\section{Introduction}\label{sec:intro}
\IEEEPARstart{M}{odern} control systems increasingly rely on wireless communication networks to connect sensors, controllers, and actuators. While this paradigm enables remote operation, reduces cabling requirements, simplifies deployment, and enhances system flexibility, it also exposes control loops to communication-induced imperfections and non-ideal transmission effects \cite{Hespanha}. Notably, network-induced imperfections, such as packet dropouts and delays, can fundamentally disrupt closed-loop behavior, as the actuator may receive a signal different from the one most recently computed by the controller. Beyond performance and stability considerations, safety is often the foremost requirement in many real-world applications. Safety constraints define the admissible operating limits of system states and control inputs, and their violation can result in hazardous operation, costly downtime, or irreversible damage \cite{adimoolam}.

Over the past decades, {control barrier certificates} (CBCs) have emerged as a powerful framework for the synthesis of safety controllers for stochastic systems, enabling the formal enforcement of safety constraints \cite{prajna2004safety}. In analogy with Lyapunov functions, a CBC is designed to satisfy prescribed conditions on the certificate itself and on its evolution along the system trajectories. By choosing an appropriate initial level set associated with a given set of initial states, CBCs separate unsafe regions from all admissible system evolutions, thereby yielding (probabilistic) safety guarantees. CBC-based techniques are widely used for formal verification and controller synthesis for both deterministic and stochastic dynamical systems ~\cite{borrmann2015control,ames2019control,santoyo2021barrier,nejati2024context,ZAKER2026113082,lavaei2024scalable,lavaei2022automated}.

While CBCs have been widely used to enforce safety constraints in stochastic control systems, most existing approaches assume uninterrupted access to feedback information. In practice, however, communication delays and packet losses can disrupt the feedback loop, causing sensor measurements to arrive late and control inputs to be unavailable at the actuator when needed, thereby hindering the timely computation and application of control actions. Several studies have investigated safety analysis for time-delayed control systems \cite{Amesdelay2023,REN2022,akbarzadeh2026safety}; however, they generally assume reliable feedback channels and do not account for packet losses in the communication links. In addition, while prior studies typically model the delay as an intrinsic component of the plant dynamics, in this work it stems from network-induced communication imperfections. The studies in~\cite{Akbarzadeh_Packet, Akbarzadeh_CoDIT} assume the absence of communication delays and consider only packet dropouts, while restricting attention to linear system dynamics. As a result, these methods are not suitable for settings in which communication delays and packet losses jointly affect controller synthesis and closed-loop performance, particularly when the plant exhibits polynomial dynamics, as considered in this paper.

Communication-induced imperfections, particularly packet losses and transmission delays, have been studied extensively in~\cite{4118454,4118476, ip-cta_20050178}. Several lines of work aim to improve reliability at the communication layer. For instance, the Glossy architecture in~\cite{5779066} exploits synchronous flooding to improve communication robustness, while other studies investigate the impact of communication imperfections on estimation accuracy and control performance. In particular, results on Kalman filtering with intermittent observations in~\cite{6004816} demonstrate that even relatively low rates of observation dropouts can significantly degrade estimation quality and, consequently, control performance. Motivated by these challenges, a variety of control strategies have been developed to mitigate the effects of communication imperfections. These include linear quadratic Gaussian control under delay-dependent constraints~\cite{8405590}, continuous-time stabilization approaches that model packet dropouts and transmission delays as time-varying input delays~\cite{yu2005stabilization}, and frameworks that explicitly capture traffic-correlated delays and packet losses~\cite{9462479}. Comprehensive overviews of networked control systems subject to communication imperfections can be found in~\cite{Hespanha}. Related works have also highlighted the importance of communication timeliness, investigating stability under repeated deadline misses~\cite{maggio_et_al} and over wireless networks~\cite{Trimpe-stability1,Trimpe-stability2}. Despite this rich body of work, the combined effects of communication delays and packet losses have not been investigated in the context of safety controller synthesis for stochastic networked control systems with polynomial dynamics, which is the focus of this work.

A related body of literature also addresses communication-aware and cross-layer design in networked control systems. In this respect, \cite{mamduhi2021delay} investigates delay-sensitive joint control, sampling, and resource management for multi-loop systems over shared networks, while \cite{mamduhi2025networkaware} develops a network-aware sampling and control framework for dynamic networks with memory. Building on these advances, \cite{klugel2019aoi} shows that several estimation objectives can be formulated through age-of-information penalty functions, employing threshold-based optimal transmission over a lossy single link. Furthermore, recent studies also examine wireless control under channel modeling and packet-dropout compensation \cite{ zacchialun2025optimal,zacchialun2025wireless}. Additionally, \cite{jungers2018observability} analyzes structural properties of linear systems subject to constrained data losses, including observability, controllability, and stabilizability. Despite accounting for communication delays, packet losses, and channel-awareness mechanisms with increasing sophistication, these studies primarily focus on \emph{linear} systems and address objectives such as performance optimization, regulation, and structural analysis, rather than the synthesis of \emph{safety} controllers for stochastic \emph{polynomial} systems.

A further related direction concerns estimation, stabilization, and delay compensation over unreliable links. Specifically, \cite{dey2014remote} studies remote estimation under packet loss, quantization, and measurement noises, while \cite{leong2008kalman} shows that Kalman smoothing under random packet loss can improve probabilistic estimation performance. Beyond estimation, \cite{maass2020stochastic} analyzes nonlinear wireless networked control systems over multiple lossy channels and derives sufficient conditions for \(\mathcal{L}_p\)-stability-in-expectation. Event-triggered sampling methods have also been proposed to mitigate communication-induced imperfections, such as delays and packet losses, thereby improving control performance in both single-loop and multi-loop linear networked control systems \cite{BALAGHII201858,QU2015974,mamduhi:MTNS2014}. More recently, \cite{dhullipalla2025event} develops an event-triggered predictor-based controller for linear networked systems with known input delays. Despite these valuable contributions, they do not address \emph{safety controller synthesis} for \emph{stochastic polynomial} systems operating under the \emph{combined effects} of uplink delays and downlink packet dropouts.

Packetized predictive control has emerged as an effective strategy for mitigating packet dropouts by storing a sequence of future control inputs at the actuator and applying them when new control packets are unavailable. For instance, \cite{QuevedoNesic2010NOLCOS} establishes stochastic stability of nonlinear systems over erasure channels without acknowledgments using a Lyapunov-based analysis. The work in \cite{QuevedoNesic2011ISS} further establishes input-to-state stability for disturbed and constrained nonlinear systems subject to bounded packet-loss bursts through appropriate controller design and tuning. For stochastic linear time-invariant systems subject to packet loss and bit-rate constraints, the integration of packetized predictive control with entropy-coded dithered lattice quantization leads to a Markov jump linear system representation, which enables tractable analysis of stability and performance trade-offs~\cite{QuevedoOstergaardNesic2011BitRate}. The proposed framework draws inspiration from the packetized predictive control paradigm and extends it to the synthesis of safety controllers with formal guarantees in the presence of network-induced communication imperfections.

\begin{figure*}[t!]
	\centering
	\includegraphics[width=0.95\textwidth]{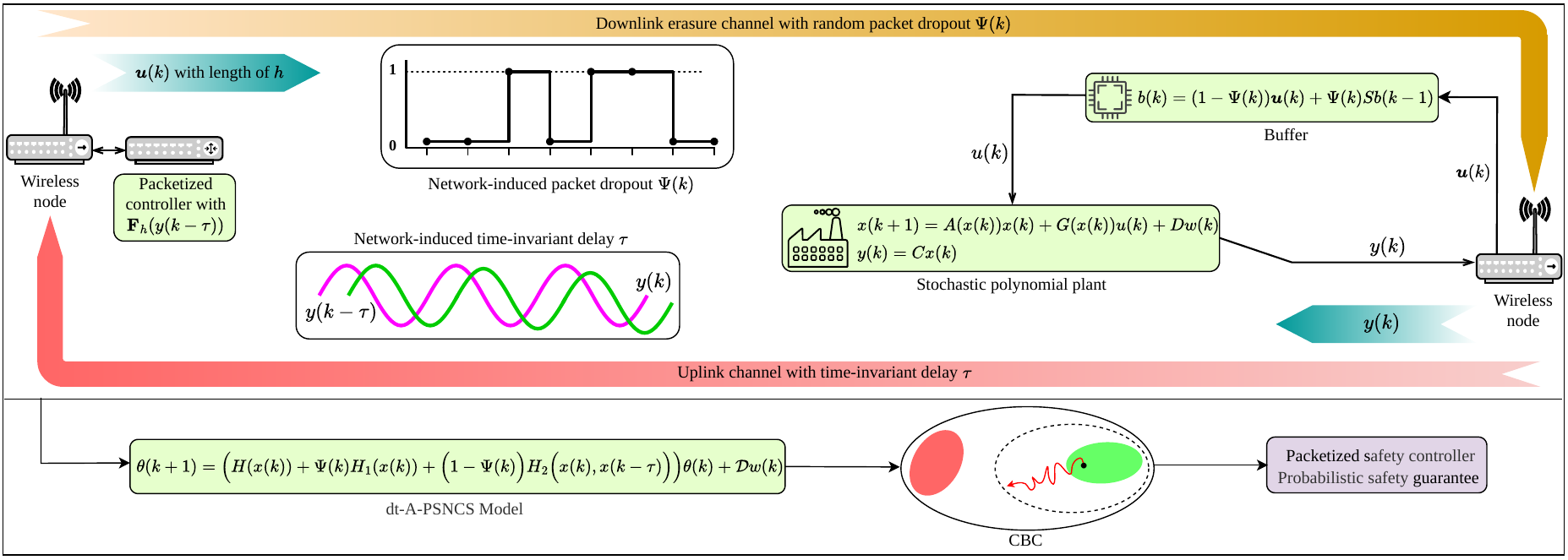}
	\caption{The architecture of the proposed framework for packetized safety controller synthesis in networked control systems. The packetized controller communicates with the plant over imperfect communication channels subject to an uplink time-invariant delay $\tau$ and packet losses in the downlink channel. Packet losses are modeled by a Bernoulli random variable $\Psi(k)$, while $w(k)$ denotes the plant process noise. In the CBC illustration, the green and red regions represent the initial and unsafe sets, respectively, while the dashed curve denotes the barrier level set associated with the initial set.
	}
	\label{fig1}
\end{figure*}

\textbf{Key contributions.} Motivated by the aforementioned challenges, this work presents a formal framework for synthesizing a packetized safety control (PSC) for discrete-time polynomial stochastic networked control systems (dt-PSNCS) under communication constraints. The framework specifically addresses delays in the plant-to-controller (uplink) channel and probabilistic packet loss in the controller-to-actuator (downlink) channel (see Fig.~\ref{fig1}). Communication imperfections are explicitly incorporated by jointly accounting for uplink delays and downlink packet losses, where the latter are modeled through an erasure channel governed by a Bernoulli random process. A buffer collocated with the plant is introduced to implement the PSC mechanism, whereby the controller transmits a packet containing a sequence of $h$ candidate control inputs, corresponding to a prediction horizon of length $h$, which are stored in the buffer and sequentially applied during packet dropouts. We further provide a guideline for selecting $h$ by deriving an upper bound on the probability of observing a consecutive dropout burst of length $h$ over the finite safety horizon. The plant and buffer states are then augmented into a unified state-space representation that captures the evolution of both the plant state and the applied control input in the presence of communication imperfections. Building upon this augmented framework, CBCs are employed to synthesize safety controllers and establish probabilistic safety guarantees through a set of SOS conditions.

A limited subset of this work was recently presented in \cite{Akbarzadeh_Networked}, with the framework developed in this paper extending those results in two key directions. First, we move beyond the restrictive class of linear systems considered in \cite{Akbarzadeh_Networked} and develop a safety framework for discrete-time stochastic \emph{polynomial} systems, thereby accommodating a richer class of nonlinear dynamics. Second, while \cite{Akbarzadeh_Networked} also focuses on safety analysis under communication delays, the present work introduces a PSC synthesis framework that explicitly addresses downlink packet losses in addition to uplink delays. This enables the synthesis of safety controllers for stochastic polynomial networked control systems operating over unreliable communication channels. Furthermore, unlike the approach in \cite[Section III-A]{Akbarzadeh_Networked}, which relies on a conservative treatment of the bilinear terms arising from products of decision variables, the proposed formulation eliminates this source of conservatism, leading to less conservative synthesis conditions and improved controller design capabilities.

\textbf{Organization.} The remainder of the article is organized as follows. Section~\ref{sec:Description} introduces the networked control system and communication model while presenting the related definitions and concepts. Section~\ref{sec: CBC} develops the CBC framework, establishes a finite-horizon probabilistic safety guarantee, reformulates the resulting safety conditions as an SOS optimization problem, and presents the corresponding PSC synthesis algorithm. Section~\ref{sec: Case} validates the proposed framework through three case studies, while Section~\ref{sec: Conclusion} concludes the paper.

\textbf{Notation.} The sets of real numbers, positive real numbers, and non-negative real numbers are denoted by $\mathbb{R}$, $\mathbb{R}^{+}$, and $\mathbb{R}_{0}^+$, respectively, while $\mathbb{N}$ and $\mathbb{N}^{+}$ denote the sets of non-negative and positive integers. For two sets $Q$ and $V$, $Q \times V$ denotes their Cartesian product, and for $n \in \mathbb{N}^+$, $Q^{n+1}$ represents the $(n+1)$-fold Cartesian product of $Q$ with itself. For $S \subseteq Q$, $Q \backslash S$ denotes the set of elements of $Q$ not contained in $S$, while $\emptyset$ represents the empty set. For a system $\Delta$ and a property $\Upsilon$, $\Delta \models \Upsilon$ indicates that $\Delta$ satisfies $\Upsilon$. Given $N$ vectors $x_i \in \mathbb{R}^{n_i}$, $x\!=\![x_1;\dots;x_N]$ denotes the concatenated column vector, $\mathbf{0}_{n}$ represents the zero vector of length $n$, and $\Vert x \Vert$ indicates the Euclidean norm of $x \in \mathbb{R}^n$. The transpose of a matrix $A$ is denoted by $A^\top$, the identity matrix of dimension $n$ by $\mathbf{I}_n$, and the zero matrix of dimension $n \times m$ by $\mathbf{0}_{n\times m}$. The trace of $A \in \mathbb{R}^{N \times N}$ with diagonal elements $(a_1, \dots , a_N)$ is $\mathsf{Tr}(A)\!=\!\sum_{i=1}^N a_i$. For a symmetric matrix $P$, $P \succ 0$ ($P \succeq 0$) denotes that $P$ is positive definite (positive semi-definite), and $(\star)$ represents the entry obtained by transposing the corresponding symmetric position. The symbol $\mathcal{N}(\mu,\Sigma)$ denotes the normal distribution with mean $\mu$ and covariance matrix $\Sigma$.

\section{Networked Control System with Communication Model}\label{sec:Description}
\subsection{Discrete-Time Polynomial Stochastic Plant Model}

A discrete-time stochastic polynomial system is considered as the primary plant model throughout this work.
\begin{definition}[\textbf{Plant Model}]\label{Plant_model}
	We consider a discrete-time polynomial stochastic plant model described via
	\begin{equation}\label{Plant}
		\begin{cases}
			x(k+1) = A(x(k))x(k) + G(x(k))u(k) + Dw(k),\\[2pt]
			y(k)   = Cx(k), \quad k \in \mathbb{N},
		\end{cases}
	\end{equation}
	where \(x\in\mathbb{X}\subset \mathbb{R}^n\), \(u\in\mathbb{U}\subset \mathbb{R}^m\), and \(y\in\mathbb{Y}\subset \mathbb{R}^q\) denote the state, input, and output, respectively, with \(\mathbb{X}\), \(\mathbb{U}\), and \(\mathbb{Y}\) being Borel spaces of admissible values. Moreover, $A:\mathbb{X}\to\mathbb{R}^{n\times n}, G:\mathbb{X}\to\mathbb{R}^{n\times m}$ are polynomial matrix-valued functions, while \(C\in\mathbb{R}^{q\times n}\) and \(D\in\mathbb{R}^{n\times n}\) are constant matrices. The process noise sequence $\{w(k)\}_{k\in\mathbb{N}}$, with $w(k)\in\mathbb{R}^n$, is assumed to be independent and identically distributed (\emph{i.i.d.}), with $w(k)\sim\mathcal{N}(\mathbf{0}_n,\mathbf{I}_n)$ for all $k\in\mathbb{N}$. $\hfill\square$
\end{definition}
We now formalize the communication model between the sensor, the remote controller, and the plant actuator (cf.~Fig.~\ref{fig1}). Specifically, the uplink channel, which transmits output measurements from the sensor to the controller, is assumed to be subject to a time-invariant delay. The downlink channel, which conveys control packets from the controller to the actuator, is modeled as an erasure channel subject to random packet losses. Together, these uplink and downlink constraints define the communication architecture considered in this work. The two channels are affected by different communication impairments and therefore exhibit distinct characteristics. In typical wireless networked control systems, sensor nodes are often located closer to the base station, making uplink communication more susceptible to delays arising from data traffic management while generally experiencing fewer packet losses. In contrast, downlink transmissions are often directed to receivers located farther from the base station, rendering them more vulnerable to packet losses while typically being less affected by delays. Consequently, modeling the uplink as delay-dominated and the downlink as loss-dominated captures the predominant communication challenges encountered in each channel (\emph{e.g.,} \cite{BENGTSSON,5G-mmWave}).

The corresponding erasure-channel model for the downlink communication link is presented in the next subsection.

\subsection{Downlink Erasure Channel Model}\label{Erasure}
This framework considers a clock-driven, Ethernet-like communication network connecting the controller and the plant, in which all information is transmitted through time-stamped packets. Since transmission errors, network congestion, power limitations, and communication distance can lead to packet dropouts, the downlink channel is modeled as an erasure channel operating synchronously with the sampling rate of the plant in~\eqref{Plant}. The packet-loss process is modeled by a Bernoulli random variable $\Psi(k)$, $k\in\mathbb{N}$, defined as follows:
\begin{equation*}
	\Psi(k) :=
	\begin{cases}
		1, & \text{if a packet dropout occurs at time } k,\\[2pt]
		0, & \text{otherwise.}
	\end{cases}
\end{equation*}
We assume that the sequence $\{\Psi(k)\}_{k\in\mathbb{N}}$  is \emph{i.i.d.} with
\begin{equation}\label{eq:bernoulli}
	\PP\big[\Psi(k)=1\big]=p,\qquad \PP\big[\Psi(k)=0\big]=1-p,
\end{equation}
where $p\in(0,1)$ denotes the packet-dropout probability.

\subsection{Packetized Control and Buffering Mechanism}
While the erasure channel model captures random packet dropouts, the actuator still requires an input at every sampling instant. To maintain continuous closed-loop operation under dropouts, we utilize a plant-side buffer that stores a packetized sequence of control inputs. Inspired by \cite{QuevedoNesic2011ISS}, we adopt a PSC buffer mechanism in which, at each time instant $k$, the controller computes and transmits an input packet $\boldsymbol{u}(k)\in\mathbb{R}^{mh}$ given by
\begin{equation*}
	\boldsymbol{u}(k)=\underbrace{\big[
		\nu_k(k); \nu_k(k+1); \cdots; \nu_k(k+h-1)
		\big]}_{(a)},
\end{equation*}
where, in $(a)$, the subscript $k$ denotes the time instant at which the control inputs are computed, while the argument specifies the time instant at which each input is intended to be applied. The vectors $\nu_k(k+i-1)\in\mathbb{R}^m$, for all $i\in\{1,\ldots,h\}$, represent the $h$ candidate control inputs computed at time $k$ and intended for application over the subsequent $h$ time steps; specifically, $\nu_k(k)$ is intended for application at time $k$, $\nu_k(k+1)$ at time $k+1$, and so on, up to $\nu_k(k+h-1)$ at time $k+h-1$. Upon successful reception of a control packet at time $k$, the actuator-side buffer is updated with the newly received sequence of planned inputs; thereafter, the buffer releases one input at each time step until another packet is successfully received, at which point its contents are overwritten by the latest sequence.

Formally, the buffer state is denoted by $b(k)\in\mathbb{R}^{mh}$ and evolves according to a \emph{receive-or-shift} update rule
\begin{equation}\label{eq:buffer-update}
	b(k)=\big(1-\Psi(k)\big)\boldsymbol{u}(k)+\Psi(k)\,S\,b(k-1),
\end{equation}
which can be rewritten as
\begin{equation*}
	b(k)=
	\begin{cases}
		\boldsymbol{u}(k), & \Psi(k)=0,\\[1mm]
		S\,b(k-1), & \Psi(k)=1.
	\end{cases}
\end{equation*}
The shift operator $S\in\mathbb{R}^{mh\times mh}$ is a parallel-in/serial-out left-shift matrix that advances the buffer by one block (of size \(m\)): \((i{+}1)\)-th block is shifted into the \(i\)-th position, while the final block is retained, \emph{i.e.,}
\begin{equation}\label{eq:S-matrix}
	S=
	\begin{bmatrix}
		\mathbf{0}_{m\times m} & \mathbf{I}_m & \mathbf{0}_{m\times m} & \cdots & \mathbf{0}_{m\times m}\\
		\mathbf{0}_{m\times m} & \mathbf{0}_{m\times m} & \mathbf{I}_m & \cdots & \mathbf{0}_{m\times m}\\
		\vdots & \vdots & \ddots & \ddots & \vdots\\
		\mathbf{0}_{m\times m} & \mathbf{0}_{m\times m} & \cdots & \mathbf{0}_{m\times m} & \mathbf{I}_m\\
		\mathbf{0}_{m\times m} & \mathbf{0}_{m\times m} & \cdots & \mathbf{0}_{m\times m} & \mathbf{I}_{m}
	\end{bmatrix}_{mh \times mh}\!\!\!\!\!\!\!\!\!\!\!\!\!\!\!\!\!\!\!\!.
\end{equation}
The applied input is extracted from the first block of the buffer according to
\begin{equation}
	u(k)= \Phi b(k),
	\qquad
	\Phi =\begin{bmatrix} \mathbf{I}_m & \mathbf{0}_{m \times m} & \cdots & \mathbf{0}_{m \times m} \end{bmatrix}\!\!.
	\label{eq:applied}
\end{equation}
Consequently, during a burst of dropouts, the buffer continues to shift and the actuator applies the next stored block at each step. As long as the packet-loss burst remains shorter than $h$ consecutive time steps, the buffer never becomes empty, and each lost packet is compensated by a previously stored candidate input from the most recently received control packet; any unused entries are discarded when a newer packet is successfully received. If a packet-loss burst persists for $h$ or more consecutive time steps, the buffer eventually exhausts all predicted inputs and, under the shift operator $S$ in \eqref{eq:S-matrix}, reaches its final block, causing the actuator to repeatedly apply the corresponding control input until a new packet is successfully received. 

The following example illustrates the role of the projection matrix $\Phi$ in \eqref{eq:applied}, which extracts the applied control input from the buffer, and the shift operator $S$ in \eqref{eq:S-matrix}, which updates the buffer during consecutive packet dropouts to ensure that a valid control input remains available until a new packet is successfully received.

\begin{example}[\textbf{Buffer Operation}]
	Consider a scalar input with prediction horizon $h=3$. Suppose that at time $k$, a packet
	$\boldsymbol{u}(k)=[\nu_k(k);\,\nu_k(k+1);\,\nu_k(k+2)]$ is successfully received. Then, the buffer is loaded as
	\begin{equation*}
		b(k)=\begin{bmatrix} \nu_k(k)\\ \nu_k(k+1)\\ \nu_k(k+2)\end{bmatrix}\!\!,
		\quad
		u(k)=\Phi b(k)= \nu_k(k).
	\end{equation*}
	If no new packet arrives at time $k{+}1$, the buffer updates according to the left-shift rule as
	\begin{align*}
		b(k{+}1) &=S\,b(k)=\begin{bmatrix} \nu_k(k+1)\\ \nu_k(k+2)\\ \nu_k(k+2)\end{bmatrix}\!\!,\\
		u(k{+}1) &=\Phi b(k{+}1)= \nu_k(k+1).
	\end{align*}
	If another dropout occurs at time $k{+}2$, one obtains
	\begin{align*}
		b(k{+}2) &= S\,b(k{+}1)\!=\!\begin{bmatrix} \nu_k(k+2)\\\nu_k(k+2)\\\nu_k(k+2)\end{bmatrix}\!\!,\\
		u(k{+}2)&=\Phi b(k{+}2)= \nu_k(k+2).
	\end{align*}
	Therefore, after $h-1$ consecutive packet dropouts, the buffer reaches its final stored block under the shift operator $S$ in \eqref{eq:S-matrix}, causing $u(\cdot)$ to repeatedly apply the corresponding control input until a new packet is successfully received.$\hfill\square$
\end{example}

\subsection{Networked Control System Model}\label{subsec:ncs-model}
Combining the plant dynamics in~\eqref{Plant} with the buffer update rule in~\eqref{eq:buffer-update},
the resulting dt-PSNCS can be formalized as follows. 
\begin{definition}[\textbf{dt-PSNCS}]\label{def: dt-SNCS}
	A discrete-time polynomial stochastic networked control system is defined by
	\begin{equation}\label{NCS}
		\Lambda :=
		\begin{cases}
			x(k{+}1) = A(x(k))x(k) + G(x(k))\Phi b(k) + Dw(k),\\[0.3em]
			y(k)     = Cx(k),\\[0.3em]
			b(k)     = (1-\Psi(k))\boldsymbol{u}(k)
			+ \Psi(k)S\,b(k{-}1),
		\end{cases}
	\end{equation}
	where at \(k=0\), we set \(b(-1)=\mathbf{0}_{mh}\), and define the initial state
	history as $\mathbf{x}(0) := [x(0);x(-1);\ldots;x(-\tau)].$$\hfill\square$
\end{definition}
Due to the network-induced uplink delay $\tau$, the controller accesses the most recent measurement $y(k-\tau)$, based on which it constructs a packet of $h$ candidate control inputs using the following output-feedback control law
\begin{equation}\label{eq:ncs-kappa}
	\boldsymbol{u}(k)=\mathbf{F}_h(y(k{-}\tau))\,y(k{-}\tau),\quad \forall k\in\mathbb{N},
\end{equation}
where $\boldsymbol{u}(k)=\big[ \nu_k(k);\ldots;\nu_k(k+h-1)\big]\in\mathbb{R}^{mh}$ and
\begin{equation*}
	\mathbf{F}_h(y(k{-}\tau)):= \big[F_0(y(k{-}\tau));\ldots;F_{h-1}(y(k{-}\tau))\big]\in \mathbb{R}^{mh\times q}
\end{equation*}
is a polynomial matrix to be designed. 

We augment the current plant state $x(k)$, the delayed state history $\mathbf{x}_d(k)$, and the previous buffer content $b(k{-}1)$ into the augmented state
\begin{equation}\label{eq:ncs-theta}
	\theta(k):=[x(k);\mathbf{x}_d(k);b(k{-}1)]\in \Theta \subseteq \mathbb{R}^{\,n(\tau+1)+mh},
\end{equation}
where $\kappa := n(\tau+1)+mh,$ and $\mathbf{x}_d(k):=[x(k{-}1);\cdots;x(k{-}\tau)] \in\mathbb{R}^{n\tau},$ and
$\Theta := \mathbb{X}^{\tau+1}\times\mathbb{U}^h.$ Hence, using~\eqref{NCS}--\eqref{eq:ncs-kappa}, the augmented dynamics are given
by
\begin{equation}\label{eq:ncs-theta-update}
	\theta(k{+}1)
	\!=\!
	\begin{bmatrix}
		A(x(k))x(k)+G(x(k))\Phi b(k)+Dw(k)\\
		x(k)\\
		\vdots\\
		x(k{-}\tau{+}1)\\
		b(k)
	\end{bmatrix}\!\!.
\end{equation}
Equivalently,~\eqref{eq:ncs-theta-update} can be written in the form of a discrete-time augmented polynomial stochastic networked control system (dt-A-PSNCS) as 
\begin{align}\notag
	\Delta\!:\theta(k{+}1) &=\! \! \big[H(x(k)) \!+\! \Psi(k)H_1(x(k)) \\\label{eq:theta-affine}  &+ (1-\Psi(k))H_2\big(x(k),x(k{-}\tau))\big]\theta(k) \!+\! \mathcal{D}w(k),
\end{align}
where $H(x(k)),\,H_1(x(k)),\,H_2(x(k),x(k-\tau))
\in
\mathbb{R}^{\kappa\times\kappa}.$
	To construct such matrices, we first introduce
	\begin{equation*}
		E_{1}
		:=
		\begin{bmatrix}
			\mathbf{I}_n\\[1mm]
			\mathbf{0}_{n(\tau-1)\times n}
		\end{bmatrix}\!\!,
		\qquad
		E_{2}
		:=
		\begin{bmatrix}
			\mathbf{0}_{n \times n(\tau-1)} & \mathbf{0}_{n \times n} \\ \mathbf{I}_{n(\tau-1)} & \mathbf{0}_{n(\tau-1) \times n}
		\end{bmatrix}\!\!,
	\end{equation*}
	using which the nominal part of the augmented
	dynamics, independent of the buffer $b(k)$ and the packetized controller $	\boldsymbol{u}(k)$,
	is represented by
	\begin{equation}\label{eq:H}
		H(x(k))
		=
		\begin{bmatrix}
			A(x(k))
			& \mathbf{0}_{n\times n\tau}
			& \mathbf{0}_{n\times mh}\\[1mm]
			E_{1}
			& E_{2}
			& \mathbf{0}_{n\tau\times mh}\\[1mm]
			\mathbf{0}_{mh\times n}
			& \mathbf{0}_{mh\times n\tau}
			& \mathbf{0}_{mh\times mh}
		\end{bmatrix}\!\!.
	\end{equation}
	The contribution associated with packet dropouts (\emph{i.e.,} $\Psi(k)=1$) is captured by $H_1(x(k))$, which describes the evolution of the closed-loop system when no new control packet is received and the actuator-side buffer is updated according to the shift matrix $S$:
	\begin{equation}\label{eq:H1}
		H_1(x(k))
		=
		\begin{bmatrix}
			\mathbf{0}_{n\times n(\tau + 1)}
			& G(x(k))\Phi S\\[1mm]
			\mathbf{0}_{n\tau\times n(\tau + 1)}
			& \mathbf{0}_{n\tau\times mh}\\[1mm]
			\mathbf{0}_{mh\times n(\tau + 1)}
			& S
		\end{bmatrix}\!\!.
	\end{equation}
	The successful-transmission component of the augmented dynamics (\emph{i.e.,} $\Psi(k)=0$) is captured by the matrix $H_2(x(k),x(k{-}\tau))$, defined as
	\begin{align}\notag
		&H_2(x(k),x(k{-}\tau))\\\label{eq:H2}
		&=
		\begin{bmatrix}
			\mathbf{0}_{n\times n\tau}
			&
			G(x(k))\Phi \mathbf{F}_h(Cx(k{-}\tau)) C
			&
			\mathbf{0}_{n\times mh}\\[1mm]
			\mathbf{0}_{n\tau\times n\tau}
			&
			\mathbf{0}_{n\tau\times n}
			&
			\mathbf{0}_{n\tau\times mh}\\[1mm]
			\mathbf{0}_{mh\times n\tau}
			&
			\mathbf{F}_h(Cx(k{-}\tau)) C
			&
			\mathbf{0}_{mh\times mh}
		\end{bmatrix}\!\!,
	\end{align}
   which can be factorized with respect to the
	polynomial controller gain \(\mathbf{F}_h\) as
	\begin{equation}\label{eq:H2-decomposition}
		H_2(x(k),x(k{-}\tau))
		=
		\bar{G}(x(k))\mathbf{F}_h(Cx(k{-}\tau))\bar C,
	\end{equation}
	where \(\bar C\in\mathbb{R}^{q\times\kappa}\) selects the delayed measured
	output from the augmented state, \emph{i.e.,}
    \begin{align} \label{eq: Cbar}
    	\bar C = \begin{bmatrix} \mathbf{0}_{q\times n\tau} & C & \mathbf{0}_{q\times mh} \end{bmatrix}\!\!, \quad \bar C\theta(k)=Cx(k-\tau),
    \end{align}
	whereas \(\bar{G}(x(k))\in\mathbb{R}^{\kappa\times mh}\) characterizes how the successfully received input sequence enters the augmented dynamics:
	\begin{equation*}
		\bar{G}(x(k))
		=
		\begin{bmatrix}
			G(x(k))\Phi\\[1mm]
			\mathbf{0}_{n\tau\times mh}\\
			\mathbf{I}_{mh}
		\end{bmatrix}\!\!.
\end{equation*}
	Since the process noise enters the dynamics only through the plant state equation, its contribution to the augmented-state dynamics is represented by
\begin{equation*}
	\mathcal{D}
	=
	\begin{bmatrix}
		D\\[1mm]
		\mathbf{0}_{(n\tau+mh)\times n}
	\end{bmatrix}\!\!.
\end{equation*}
We represent dt-A-PSNCS in~\eqref{eq:theta-affine} using the tuple $\Delta=(\mathbb{X}, \mathbb{U}, \mathbb{Y}, A, G, C,D,\tau, p)$. For any initial augmented state $\theta(0) \in \Theta$, the random sequence $\theta_s$ satisfying \eqref{eq:theta-affine} under an input signal $u(\cdot)$ is called the solution process of $\Delta$. We now present the formal definition of safety for dt-A-PSNCS.
\begin{definition}[\textbf{Safety Property}]\label{safety}
	Given a dt-A-PSNCS $\Delta=(\mathbb{X}, \mathbb{U}, \mathbb{Y}, A, G, C,D, \tau, p)$ in~\eqref{eq:theta-affine}, consider a safety specification $\Upsilon=(\mathbb X_a, \mathbb X_b, \mathcal{T})$, where $ \mathbb X_a, \mathbb X_b \allowbreak  \subset \mathbb X$ are the initial and unsafe sets of the plant model in \eqref{Plant}, respectively, where $\mathbb X_a \cap  \mathbb X_b=\emptyset$. The dt-A-PSNCS is said to be safe within the time horizon $\mathcal{T} \in \mathbb{N}$, denoted by $\Delta \models \Upsilon$, if for every initial state $\theta(0) \in \Theta_a := \mathbb X_a^{\tau+1} \times \mathbb U^h$, the corresponding trajectory satisfies $\theta(k) \notin \Theta_b := \mathbb X_b \times (\mathbb X \setminus \mathbb X_b)^\tau \times \mathbb U^h$ for all $k \in \{1,\dots,\!\mathcal{T}\}$. Since trajectories of $\Delta$ evolve stochastically, the aim is to formally compute $\PP \{\Delta \models \Upsilon\} \geq 1-\rho$, where $\rho \in(0,1]$ denotes the safety violation probability.$\hfill\square$
\end{definition}

It is clear that $\mathbb{X}_a \cap \mathbb{X}_b = \emptyset$, implies the disjointness of the augmented sets, $\Theta_a = \mathbb X_a^{\tau+1} \times \mathbb U^h$ and $\Theta_b = \mathbb X_b \times (\mathbb X \setminus \mathbb X_b)^\tau \times \mathbb U^h$, \emph{i.e.,} $\Theta_a\cap\Theta_b=\emptyset$. Indeed, every $\theta\in\Theta_a$ has its current state in $\mathbb{X}_a$, whereas every $\theta \in \Theta_b$ has its current state in $\mathbb{X}_b$; since these two sets are disjoint, no augmented state can belong to both $\Theta_a$ and $\Theta_b$.
\begin{remark}
	A state history trajectory can avoid the set $\mathbb{X}_b^{\tau+1}$ while still violating safety, as this set only characterizes histories where \emph{all} $\tau + 1$ components of the history are in the unsafe set $\mathbb{X}_b$. The set $\mathbb{X}_b \times (\mathbb{X} \setminus \mathbb{X}_b)^\tau$, however, considers the case where the current state enters the unsafe set while all previous $\tau$ states were safe. In other words, it captures the instances in which the first safety violation occurs. Since the initial state history is entirely safe (\emph{i.e.,} in $\mathbb{X}_a^{\tau+1}$ with $\mathbb{X}_a \cap \mathbb{X}_b = \emptyset$), guaranteeing that the augmented trajectory never enters $\Theta_b$ prevents this first violation from occurring, which implies that the system remains outside $\mathbb{X}_b$ for all time.$\hfill\square$
\end{remark}
To guide the selection of the control input prediction horizon $h$, we present Lemma~\ref{lem:run_bernoulli}, which provides an upper bound on the probability of observing at least one run of $h$ consecutive packet dropouts within the finite safety horizon $\mathcal{T}$ under the i.i.d. Bernoulli erasure channel model (cf. Subsection~\ref{Erasure}). This result offers a practical criterion for choosing $h$ with $h<\mathcal{T}$.
\begin{lemma}[\textbf{Input Prediction Horizon $h$}]\label{lem:run_bernoulli}
	Consider a dt-A-PSNCS in \eqref{eq:theta-affine}, where $\Psi(k)\in\{0,1\}$ denotes the downlink packet-dropout indicator at time step $k$, as defined in \eqref{eq:bernoulli}. For a finite safety horizon $\mathcal{T}\in\mathbb{N}$ and a prediction horizon $h\in\{1,\ldots,\mathcal{T}-1\}$, the event of observing at least one run of $h$ consecutive packet dropouts within the interval $\{0,1,\ldots,\mathcal{T}\}$ is given by
	\begin{align*}
		\mathcal{E}_{h,\mathcal{T}}&
		:=\Big\{\exists\, k\in\{0,\dots,\mathcal{T}-h+1\}\\
		& \text{s.t.}\
		\Psi(k)=\Psi(k+1)=\cdots=\Psi(k+h-1)=1\Big\}.
	\end{align*}
	Then
	\begin{equation}
		\PP(\mathcal{E}_{h,\mathcal{T}})\le (\mathcal{T}-h+2)\,p^h.
		\label{eq:run_bound_k_ET_final}
	\end{equation}
\end{lemma}
\begin{proof}
For each possible starting time step $k\in\{0,\dots,\mathcal{T}-h+1\}$, let us define
the event
\begin{equation*}
	\mathcal{E}_k:=\big\{\Psi(k)=\Psi(k+1)=\cdots=\Psi(k+h-1)=1\big\}.
	\label{eq:Ek_def_ET_final}
\end{equation*}
The occurrence of a consecutive $h$-step dropout run within $\{0,\dots,\mathcal{T}\}$
is equivalent to the union of these events, \emph{i.e.,}
\begin{equation}
	\mathcal{E}_{h,\mathcal{T}}=\bigcup_{k=0}^{\mathcal{T}-h+1} \mathcal{E}_k.
	\label{eq:union_k_ET_final}
\end{equation}
Applying Boole's inequality (union bound) to \eqref{eq:union_k_ET_final} yields
\begin{equation}
	\PP(\mathcal{E}_{h,\mathcal{T}})\le \sum_{k=0}^{\mathcal{T}-h+1}\PP(\mathcal{E}_k).
	\label{eq:union_bound_k_ET_final}
\end{equation}
Since $\Psi(k)$ is \emph{i.i.d.} and \eqref{eq:bernoulli} holds, we obtain
\begin{align}\notag
	\PP(\mathcal{E}_k)
	&= \PP\Big(\bigcap_{i=0}^{h-1}\{\Psi(k+i)=1\}\Big)
	\\	\label{eq:Ek_prob_ET_final} &= \prod_{i=0}^{h-1}\PP\!\big(\Psi(k+i)=1\big)
	= p^h.
\end{align}
As there are $\mathcal{T}-h+2$ admissible starting indices in
$\{0,\dots,\mathcal{T}-h+1\}$, substituting \eqref{eq:Ek_prob_ET_final} into
\eqref{eq:union_bound_k_ET_final} yields
\begin{equation*}
	\PP(\mathcal{E}_{h,\mathcal{T}})\le (\mathcal{T}-h+2)\,p^h,
\end{equation*}
which completes the proof. 
\end{proof}
\begin{remark}
	While the bound in \eqref{eq:run_bound_k_ET_final} is potentially conservative due to the application of the union bound, it provides a practical criterion for selecting the control input prediction horizon $h$. Specifically, it ensures that the likelihood of observing an $h$-step outage within the prescribed finite horizon $\mathcal{T}$ remains strictly below a desired tolerance.$\hfill\square$
\end{remark}

\textbf{Design Guideline.}
Lemma~\ref{lem:run_bernoulli} provides a (potentially conservative) upper bound on the probability
of observing at least one $h$-step dropout burst over the finite safety horizon
$\mathcal{T}$. Accordingly, the control input prediction horizon $h$ can be chosen as the
smallest integer $h\in\{1,\dots,\mathcal{T}-1\}$, satisfying
\begin{equation}
	(\mathcal{T}-h+2)\,p^h \le \alpha,
	\label{eq:N_design}
\end{equation}
where $\alpha \in(0,1)$. This choice ensures that the likelihood of an $h$-consecutive-dropout event
$\mathcal{E}_{h,\mathcal{T}}$ remains below a prescribed allowable outage probability (see \cite{Popovski2019URLLC, Chang2019}).
We emphasize that $\alpha$ is a designer-specified \emph{allowable outage probability}, representing an upper bound on the probability of experiencing at least one $h$-step communication outage over the interval ${0,1,\dots,\mathcal{T}}$. Its value reflects the desired level of communication robustness and the criticality of the application, with safety-critical systems typically requiring smaller values of $\alpha$.

According to~\eqref{eq:N_design}, the selected horizon $h$ depends on the allowable outage probability
	$\alpha$, the downlink dropout likelihood $p$, which is a property of the
	communication channel, and the safety horizon $\mathcal{T}$. Under the proposed design rule, a less reliable channel, with a
	larger dropout probability \(p\) or a longer safety horizon \(\mathcal T\),
	may require a larger prediction horizon \(h\) in order to satisfy the
	prescribed outage tolerance \(\alpha\). Similarly, smaller values of
	\(\alpha\) generally lead to larger values of \(h\), thereby increasing the
	computational burden through a longer prediction horizon and its associated
	buffer dimension.
\section{Controller Synthesis with Safety Guarantees}\label{sec: CBC}
This section presents the main results of the work. Specifically, we define a CBC for dt-A-PSNCS, establish a finite-horizon probabilistic safety guarantee, and derive tractable conditions for synthesizing both a CBC and its corresponding PSC through an SOS optimization program. To do so, we first introduce the notion of CBC in the following definition.
\begin{definition}[\textbf{CBC}]\label{def: CBC}
	Consider a dt-A-PSNCS $\Delta=(\mathbb{X}, \mathbb{U}, \mathbb{Y}, A, G,C,D, \tau, p)$ with the augmented state space $\Theta =\mathbb{X}^{\tau+1} \times \mathbb{U}^h$. Suppose that $\Theta_a =\mathbb{X}_a^{\tau+1} \times \mathbb{U}^h \subset \Theta$ and $\Theta_b =\mathbb{X}_b \times (\mathbb{X}\backslash\mathbb{X}_b)^{\tau} \times \mathbb{U}^h \subset \Theta$ denote the augmented initial and unsafe sets, respectively, with $\Theta_a \cap \Theta_b=\emptyset$. Let the control input $\boldsymbol{u}(k)=\mathbf{F}_h(y(k{-}\tau))\,y(k{-}\tau)$ be given as in~\eqref{eq:ncs-kappa}. A function $\mathcal B: \Theta \to \mathbb{R}_0^+$ is called a CBC for dt-A-PSNCS if there exist constants $\eta, \gamma_a, \gamma_b \in \mathbb{R}^{+}$, with $\gamma_b > \gamma_a$, such that
	\begin{subequations}\label{eq: CBC}
		\begin{align}
			&  \:\:  \mathcal B(\theta) \leq \gamma_a, \hspace{3cm}  \forall \theta \in \Theta_{a},\label{subeq: initial}\\
			&  \:\:  \mathcal B(\theta) \geq \gamma_b, \hspace{3cm} \forall \theta \in \Theta_b, \label{subeq: unsafe}
		\end{align}  
		and $\forall \theta \in \Theta$ 
		\begin{align}\label{subeq: decreasing}
			& \EE \big[ \mathcal B(\theta(k+1)) \!\mid\! \theta(k) \big] - \mathcal B(\theta(k)) \leq  \eta,
		\end{align}
				\end{subequations}
	where $\EE$ denotes the conditional expectation with respect to the random packet-loss process  $\Psi(k)$ and the process noise  $w(k)$.$\hfill\square$
\end{definition}
Under Definition~\ref{def: CBC}, the following theorem establishes a lower bound on the probability that the system trajectories in~\eqref{eq:theta-affine} remain within the safe set over a finite horizon~\cite{Kushner}.

\begin{theorem}[\textbf{Probabilistic Safety Guarantee}]\label{Th: safety}
	Given a dt-A-PSNCS $\Delta=(\mathbb{X}, \mathbb{U}, \mathbb{Y}, A, G,C,D, \tau, p)$ in~\eqref{eq:theta-affine}, assume there exists a CBC $\mathcal{B}$ as in Definition~\ref{def: CBC}. Then, the probability that the solution process $\theta_s$ of dt-A-PSNCS, starting from any initial condition $\theta(0)  \in \Theta_a$ under an input signal $u(\cdot)$ never reaches $\Theta_b$ within a time horizon $\mathcal{T}$ is lower-bounded as
	\begin{equation}\label{bound}
		\PP \Big\{ {\theta}_s(k) \notin \Theta_b \text { for all } k \in \{1,\dots,\!\mathcal{T}\} \,\big|\, \theta(0) \Big\} \geq 1- \rho,
	\end{equation}
	with $\rho = \frac{\gamma_a + \eta \mathcal{T}}{\gamma_b}$.
\end{theorem}

\begin{remark}
	The CBC condition~\eqref{subeq: decreasing} is affected by both the uplink delay~$\tau$ and the probabilistic packet loss in the downlink $\Psi$, since these factors are explicitly incorporated into dt-A-PSNCS dynamics described in~\eqref{eq:theta-affine}. As a result, such communication imperfections directly impact the CBC $\mathcal{B}$, subsequently altering the level set values $\gamma_a$~\eqref{subeq: initial} and $\gamma_b$~\eqref{subeq: unsafe}, as well as the constant $\eta$. Together, these parameters determine the lower bound on the safety probability~$1-\rho$, as detailed in Theorem~\ref{Th: safety}.$\hfill\square$
\end{remark}

In this work, we utilize a quadratic CBC in the form of $\mathcal{B}(\theta) = \theta^\top P \theta$, where $P \succ 0$, $P \in \mathbb{R}^{\kappa \times \kappa}$. This approach enables the reformulation of condition~\eqref{subeq: decreasing} as a tractable matrix inequality, which is presented in the following theorem and constitutes the main result of this work. For notational convenience, we denote the current state by \(x := x(k)\) and the delayed state by \(x_{\tau} := x(k-\tau)\).

\begin{theorem}[\textbf{On Condition~\eqref{subeq: decreasing}}]
		\label{Th:inverse_decay}
		Consider the dt-A-PSNCS \(\Delta=(\mathbb X,\mathbb U,\mathbb Y,A,G,C,D,\tau,p)\) as 
		in~\eqref{eq:theta-affine}, with packet-dropout rate \(p\in(0,1)\) and
		time-invariant delay \(\tau\). Suppose that there exists a symmetric
		matrix \(Q\succ0\), $Q \in \mathbb{R}^{\kappa \times \kappa}$, a symmetric matrix \(M\succ0\), \(M\in\mathbb R^{q\times q}\), and
		a polynomial matrix \(Y(Cx_\tau)\in\mathbb R^{mh\times q}\) such that
		\begin{equation}\label{eq:recoverability}
			\bar C Q=M\bar C,
		\end{equation}
		with $\bar{C}$ as defined in \eqref{eq: Cbar}, and for all \((x,x_\tau)\in\mathbb X^2\),
		\begin{equation}\label{eq:inverse-LMI-compact}
			\begin{bmatrix}
				Q &&& \Omega_1^\top(x) &&& \Omega_2^\top(x,x_\tau)\\
				\star &&& Q &&& \mathbf 0_{\kappa\times\kappa}\\
				\star &&& \star &&& Q
			\end{bmatrix}
			\succeq0,
		\end{equation}
		where
		\begin{align}\notag
			\Omega_1(x)
			&= \sqrt p\,(H(x)+H_1(x))Q,\\\label{Omega}
			\Omega_2(x,x_\tau)
			&= \sqrt{1-p}\,(H(x)Q+\bar{G}(x)Y(Cx_\tau)\bar C).
		\end{align}
		Then, condition~\eqref{subeq: decreasing} holds for $\mathcal{B}(\theta) = \theta^\top P \theta$ with $P=Q^{-1},$ $\mathbf F_h(Cx_\tau)=Y(Cx_\tau)M^{-1},$ and $\eta=\mathsf{Tr}(\mathcal D^\top Q^{-1}\mathcal D).$
\end{theorem}
\begin{proof}
For brevity, let us define $\theta:=\theta(k)$, $\theta^+:=\theta(k{+}1)$, $\Psi:=\Psi(k)$, $w:=w(k)$,
$H:=H(x)$, $H_1:=H_1(x)$, and $H_2:=H_2(x,x_\tau)$ in \eqref{eq:ncs-theta}-\eqref{eq:H2}. Then
\begin{equation*}
	\theta^+ = (H+ \Psi H_1+(1-\Psi)H_2)\theta + \mathcal{D}w.
	\label{eq:theta-next}
\end{equation*}
By evaluating the conditional expectation of $\mathcal{B}$ at $k{+}1$, we obtain 
\begin{align*}
	&\EE\big[\mathcal{B}(\theta^+) \, \big|\, \theta\big]\\
	&= \EE\big[\big((H+ \Psi H_1+(1-\Psi)H_2)\theta + \mathcal{D} w\big)^\top P\\ & ~~~~~~~~\big((H+ \Psi H_1+(1-\Psi)H_2)\theta + \mathcal{D} w\big)\,\big|\,\theta\big] \\
	&= \EE\big[\theta^\top(H+ \Psi H_1+(1-\Psi)H_2)^\top P\\ &~~~~~~~~~ (H+ \Psi H_1+(1-\Psi)H_2)\theta \,\big|\,\theta\big]\\
	&~~~ + 2\EE\big[\theta^\top (H+ \Psi H_1+(1-\Psi)H_2)^\top P \mathcal{D} w \,\big|\,\theta\big]
	\\ &~~~+ \EE\big[w^\top \mathcal{D}^\top P \mathcal{D} w\big]\!.
\end{align*}
Since $\EE \big[w\mid \theta \big]=0$ and $w$ is independent of $(\theta,\Psi)$, the cross term vanishes as
\begin{equation*}
	\EE\big[\theta^\top (H+ \Psi H_1+(1-\Psi)H_2)^\top P\mathcal{D}w \,\big|\,\theta\big]=0.
\end{equation*}
Moreover, since $\EE \big[ww^\top \big]=\mathbf{I}_n$ and by the cyclic property of the trace operator, one has
\begin{align*}
	\EE\big[w^\top\mathcal{D}^\top P\mathcal{D}w\big]
	&=\mathsf{Tr}\big(\mathcal{D}^\top P\mathcal{D}\,\EE \big[ww^\top \big]\big)
	\\ &=\mathsf{Tr}\big(\mathcal{D}^\top P\mathcal{D}\big).
\end{align*}
In addition, 
\begin{align}\notag
	&\EE \big[\theta^\top(H+ \Psi H_1+(1-\Psi)H_2)^\top P\\\notag &~~~~~ (H+ \Psi H_1+(1-\Psi)H_2) \theta \,\big|\,\theta\big]\\\notag
	&= \theta^\top\EE\big[(H+ \Psi H_1+(1-\Psi)H_2)^\top P\\\label{exp1} &~~~~~ (H+ \Psi H_1+(1-\Psi)H_2) \,\big|\,\theta\big]\theta.
\end{align}	
Since $\Psi$ is Bernoulli distributed, we have $\EE\big[\Psi\big] = \EE\big[\Psi^2\big]=p,$ and the expression in \eqref{exp1} can be written as 
\begin{align*}
	&\theta^\top \!\! \Big( \EE \big[H^{\top} P H \big] \!+\! p \EE \big[H^{\top} P H_1 \big] \!+\! (1 \!-\! p) \EE[H^{\top} P H_2] \\
	& ~~~+\! p \EE \big[H_1^{\top} P H \big] \!+\! p \EE \big[H_1^{\top} P H_1 \big] \!+\! \underbrace{\EE \big[\Psi(1 \!-\! \Psi) \big] \EE \big[H_1^{\top} P H_2 \big]}_{=0} \\
	& ~~~+(1-p)\EE \big[ H_2^{\top} P H \big] +\underbrace{\EE \big[\Psi(1-\Psi) \big] \EE \big[H_2^{\top} P H_1 \big]}_{=0} \\ &~~~+(1-p) \EE \big[H_2^{\top} P H_2 \big]\Big)\theta \\
	&= \theta^\top \Big(
	H^\top P H
	+ p H^\top P H_1
	+ (1-p)H^\top P H_2
	+ p H_1^\top P H
	\\
	&\qquad
	+ p H_1^\top P H_1 + (1-p)H_2^\top P H
	+ (1-p)H_2^\top P H_2
	\Big)\theta.
\end{align*}
Then, the condition in \eqref{subeq: decreasing} can be written as
\begin{align*}
	&\EE \big[\mathcal{B}\big(\theta(k{+}1)\big)\,\big|\, \theta(k)\big] - \mathcal{B}\big(\theta(k)\big)\\
	&= \theta^\top \big( H^{\top} P H \!+\! pH^{\top} P H_1 \!+\! (1-p)H^{\top} P H_2 \!+\!  pH_1^{\top} P H\\ &~~~ + pH_1^{\top} P H_1 \!+\! (1-p)H_2^{\top} P H \!+\! (1-p)H_2^{\top} P H_2 \!-\! P \big)\theta \\ &~~~+ \underbrace{\mathsf{Tr}\big( \mathcal{D}^\top P \mathcal{D}\big)}_{\eta}.
\end{align*}
Then, to satisfy \eqref{subeq: decreasing}, it is sufficient to show that
\begin{align}\notag
	&H^{\top} P H \!+\! pH^{\top} P H_1 \!+\! (1-p)H^{\top} P H_2 \!+\! pH_1^{\top} P H \\\label{exp2} &~+ pH_1^{\top} P H_1 \!+\! (1-p)H_2^{\top} P H \!+\! (1-p)H_2^{\top} P H_2 \!-\! P \preceq 0,
\end{align}
with $\eta = \mathsf{Tr}\big( \mathcal{D}^\top P  \mathcal{D} \big)$.  	The left-hand side of~\eqref{exp2} can be grouped as
\begin{align}\notag
	&p(H+H_1)^{\top} P(H+H_1)\\\label{exp3} &~~~+(1-p)(H+H_2)^{\top} P\left(H+H_2\right)-P \preceq 0 .
\end{align}
To avoid the bilinear coupling between the decision matrix \(P\) and the
	controller-dependent matrix \(H_2\), we employ \(P=Q^{-1}\),
	\(\mathcal J_1=\sqrt p\,(H+H_1)\), and
	\(\mathcal J_2=\sqrt{1-p}(H+H_2)\). Then, condition~\eqref{exp3}
	can be reformulated as
	\begin{equation}\label{eq:J-decrease-P}
		\mathcal J_1^\top P\mathcal J_1+
		\mathcal J_2^\top P\mathcal J_2-P\preceq0 .
	\end{equation}
	Multiplying \eqref{eq:J-decrease-P} from the left and right by \(Q=P^{-1}\) results in
	\begin{equation}\label{eq:J-decrease-Q}
		(\mathcal J_1Q)^\top Q^{-1}(\mathcal J_1Q)
		+
		(\mathcal J_2Q)^\top Q^{-1}(\mathcal J_2Q)
		-Q\preceq0 .
	\end{equation}
	Using the Schur complement,~\eqref{eq:J-decrease-Q} is equivalent to
	\begin{equation*}
		\begin{bmatrix}
			Q &&& (\mathcal J_1Q)^\top &&& (\mathcal J_2Q)^\top\\
			\star &&& Q &&& \mathbf 0_{\kappa\times\kappa}\\
			\star &&& \star &&& Q
		\end{bmatrix}
		\succeq0.
	\end{equation*}
	Using \(H_2=\bar{G}(x)\mathbf F_h(Cx_\tau)\bar C\) in \eqref{eq:H2-decomposition}, one has
	\begin{equation}\label{eq:J2Q-before-recovery}
		\mathcal J_2Q
		=
		\sqrt{1-p}
		\left(HQ+\bar{G}(x)\mathbf F_h(Cx_\tau)\bar C Q\right).
	\end{equation}
	By the constraint \(\bar C Q=M\bar C\) in \eqref{eq:recoverability}, it follows that
	\begin{equation*}
		\mathbf F_h(Cx_\tau)\bar C Q
		=
		\mathbf F_h(Cx_\tau)M\bar C .
	\end{equation*}
	Defining \(Y(Cx_\tau)=\mathbf F_h(Cx_\tau)M\),~\eqref{eq:J2Q-before-recovery}
	becomes $\mathcal J_2Q=\sqrt{1-p}(HQ+\bar{G}(x)Y(Cx_\tau)\bar C).$ Therefore, condition~\eqref{eq:inverse-LMI-compact} implies~\eqref{eq:J-decrease-P}.
	In addition,
	\begin{equation*}
		\eta
		=
		\mathsf{Tr}(\mathcal D^\top P\mathcal D)
		=
		\mathsf{Tr}(\mathcal D^\top Q^{-1}\mathcal D).
	\end{equation*}
	Finally, from \(Y(Cx_\tau)=\mathbf F_h(Cx_\tau)M\), the PSC polynomial gain  is $\mathbf F_h(Cx_\tau)=Y(Cx_\tau)M^{-1},$ which completes the proof.
\end{proof}
	\subsection{Computation of CBC and PSC}
   Here, we present the following lemma to facilitate the design of a CBC and its associated PSC for dt-A-PSNCS through an SOS program.
\begin{lemma}[\textbf{SOS Program}]\label{lem:inverse_sos} 
	Let the state set $\mathbb{X}$ be defined by a vector of polynomial inequalities as $ \mathbb{X}=\left\{ x  \in \mathbb{R}^n  \mid J(x) \geq 0\right\}$. Suppose there exist matrices $Q \succ 0$, $M \succ 0$ with $	\bar C Q=M\bar C$, a state-dependent polynomial matrix $Y(Cx_\tau)$, and vectors of SOS polynomials $\mathcal{Y}(x,x_\tau)$ and $\mathcal{Y}_{\tau}(x,x_\tau)$, such that
	\begin{align}\notag
	\Gamma(x,x_\tau)
	&=
	\begin{bmatrix}
		Q &&& \Omega_1^\top(x) &&& \Omega_2^\top(x,x_\tau)\\
		\star &&& Q &&& \mathbf 0_{\kappa\times\kappa}\\
		\star &&& \star &&& Q
	\end{bmatrix}\\ \label{eq:SOS-inverse-condition} &~~~
	-
	(
	\mathcal Y(x,x_\tau)^\top J(x)
	+
	\mathcal Y_\tau(x,x_\tau)^\top J(x_\tau)
	) \mathbf{I}_{3\kappa},
\end{align}
with $\Omega_1(x)$ and $\Omega_2(x,x_\tau)$ as in \eqref{Omega}. Let \(\gamma_a\) and \(\gamma_b\) be defined as
	\begin{subequations}
\begin{align}\label{eq:gamma-a-inverse}
	\gamma_a
	&=
	\lambda_{\max}(Q^{-1})
	\max_{\theta\in\Theta_a}\|\theta\|^2 ,\\\label{eq:gamma-b-inverse}
	\gamma_b
	&=
	\lambda_{\min}(Q^{-1})
	\min_{\theta\in\Theta_b}\|\theta\|^2 .
\end{align}
\end{subequations}
If \(\gamma_b>\gamma_a\), then \(\mathcal B(\theta)=\theta^\top Q^{-1}\theta\) is a CBC and
\(\boldsymbol{u}(k)\) in \eqref{eq:ncs-kappa} is a PSC with
\begin{equation}\label{eq:eta-inverse}
	\eta=\mathsf{Tr}(\mathcal D^\top Q^{-1}\mathcal D).
\end{equation}
\end{lemma}
\begin{proof}
We first show that \eqref{eq:SOS-inverse-condition} implies \eqref{eq:inverse-LMI-compact}. As $\mathcal{Y}(x,x_\tau)$ and $\mathcal{Y}_\tau(x,x_{\tau})$ are both SOS polynomials, {we have $\mathcal{Y}^\top(x,x_\tau)J(x) +\mathcal{Y}_\tau^\top(x,x_\tau)J(x_\tau)\ge 0$} for all $(x,x_\tau) \in \mathbb{X}^2$. Given that \eqref{eq:SOS-inverse-condition} is an SOS matrix polynomial, it follows that $\Gamma(x,x_\tau)\succeq 0$. Consequently, \eqref{eq:inverse-LMI-compact} holds, which in turn implies \eqref{subeq: decreasing}.

\begin{algorithm}[t]
	\caption{CBC and PSC synthesis under communication imperfections}
	\label{Alg2}
	\begin{algorithmic}[1]
		\REQUIRE dt-A-PSNCS, \(\Upsilon=(\mathbb X_a,\mathbb X_b,\mathcal T)\), packet-dropout probability $p$, allowable outage probability \(\alpha\) in~\eqref{eq:N_design}
		\STATE Set the parallel-in/serial-out left-shift operator \(S\) in~\eqref{eq:S-matrix}
		\STATE Compute the control input prediction horizon \(h\) according to~\eqref{eq:N_design}
		\STATE Construct \(H(x)\), \(H_1(x)\), and \(H_2(x,x_\tau)\) in~\eqref{eq:H}, \eqref{eq:H1}, and~\eqref{eq:H2}
		\STATE Construct the decomposition \(H_2(x,x_\tau)=\bar{G}(x)\mathbf F_h(Cx_\tau)\bar C\)
		\STATE Solve the SOS condition in~\eqref{eq:SOS-inverse-condition} using \textsf{SOSTOOLS}~\cite{SOSTOOLS} with \textsf{Mosek}~\cite{mosek} to compute \(Q\succ0\), \(M\succ0\) and the polynomial matrix \(Y(Cx_\tau)\), subject to \(\bar C Q=M\bar C\)
		\STATE Recover \(P=Q^{-1}\) and \(\mathbf F_h(Cx_\tau)=Y(Cx_\tau)M^{-1}\)
		\STATE Compute the applied control input according to \(u(k)=\Phi b(k)\) in~\eqref{eq:applied}
		\STATE Design level sets \(\gamma_a,\gamma_b\) in~\eqref{eq:gamma-a-inverse} and~\eqref{eq:gamma-b-inverse} using the constructed \(\mathcal B(\theta)\)
		\STATE Compute \(\eta=\mathsf{Tr}(\mathcal D^\top Q^{-1}\mathcal D)\) in~\eqref{eq:eta-inverse}
		\STATE Compute the safety probability lower bound \(1-\rho=1-\frac{\gamma_a+\eta\mathcal T}{\gamma_b}\)
		\ENSURE CBC \(\mathcal B(\theta)=\theta^\top Q^{-1}\theta\), PSC \(\boldsymbol{u}(k)\) in \eqref{eq:ncs-kappa}, and probabilistic safety guarantee in~\eqref{bound}
	\end{algorithmic}
\end{algorithm}

\begin{table*}[t]
	\makeatletter
	\long\def\@makecaption#1#2{%
		\vskip\abovecaptionskip
		\noindent \textbf{#1.} #2\par
		\vskip\belowcaptionskip}
	\makeatother
	\centering
	\caption{Summary of the results obtained for the case studies, where $h$ denotes the control input prediction horizon, $\alpha$ the allowable outage probability, $\tau$ the time-invariant uplink delay, $\mathsf{deg}$ the polynomial degree of the system dynamics, $n$ the state dimension, $p$ the downlink packet-dropout rate, $\gamma_a$ and $\gamma_b$ the level-set parameters associated with the initial and unsafe sets, $1-\rho$ the lower bound on the safety probability, $\eta$ the parameter defined in \eqref{eq:eta-inverse}, and $\mathcal{T}$ the finite safety horizon.}
	\label{tab:system-configurations}
	\vspace{2mm}
	\renewcommand{\arraystretch}{1.15}
	\begin{tabular}{@{}l|ccccccccccc@{}}
		\toprule
		Case study& $h$ & $\alpha$ & $\tau$ & $\mathsf{deg}$ & $n$ & $p$ & $\gamma_a$ & $\gamma_b$ & $1-\rho$ & $\eta$ & $\mathcal{T}$ \\
		\midrule
		\textbf{Van der Pol oscillator}
		& $2$ &   $ 0.090$ & $2$ &  $3$& $3$ & $0.015$ & $13.34$ &$427.35$ & $93\%$ & $0.07$ &200 \\\midrule
		
		\textbf{Jet engine compressor}
		& 3 & $0.054$ & $2$ & $3$ & $2$ & $0.04$ &  $4.06 \times 10^{5}$  & $1.44 \times 10^{7}$  & $90\%$ & $9.39 \times 10^{3}$ & $100$ \\\midrule
		
		\textbf{Academic system}
		& 3 & $ 0.0127$ & $2$ & $2$ & $2$ & $0.03$ & $6.06 \times 10^{6}$ & $ 9.34 \times 10^{7}$ & $93\%$ &  $ 2.05 \times 10^{3}$&$200$ \\
		\bottomrule
	\end{tabular}
\end{table*}

\noindent We now show conditions~\eqref{subeq: initial} and~\eqref{subeq: unsafe} are satisfied, as well.
Since \(P=Q^{-1}\), one has
\(\mathcal B(\theta)=\theta^\top Q^{-1}\theta\). For any
\(\theta\in\Theta_a\), 
\begin{equation*}
	\mathcal B(\theta) \!=\!
	\theta^\top Q^{-1}\theta
	\!	\leq \!
	\lambda_{\max}(Q^{-1})\|\theta\|^2
	\! \leq \!
	\lambda_{\max}(Q^{-1})
	\max_{\theta\in\Theta_a}\|\theta\|^2 .
\end{equation*}
Hence, condition~\eqref{subeq: initial} holds with
\begin{equation*}
	\gamma_a =
	\lambda_{\max}(Q^{-1})
	\max_{\theta\in\Theta_a}\|\theta\|^2 .
\end{equation*}
Similarly, for any \(\theta\in\Theta_b\),
\begin{equation*}
	\mathcal B(\theta) \!=\!
	\theta^\top Q^{-1}\theta
	\! \geq \!
	\lambda_{\min}(Q^{-1})\|\theta\|^2
	\! \geq \!
	\lambda_{\min}(Q^{-1})
	\min_{\theta\in\Theta_b}\|\theta\|^2 .
\end{equation*}
Therefore, condition~\eqref{subeq: unsafe} is met with
\begin{equation*}
	\gamma_b
	=
	\lambda_{\min}(Q^{-1})
	\min_{\theta\in\Theta_b}\|\theta\|^2 .
\end{equation*}
Since \(\gamma_b>\gamma_a\), then \(\mathcal B(\theta)=\theta^\top Q^{-1}\theta\) is a CBC and
\(\boldsymbol{u}(k)\) is its PSC with the polynomial gain \(\mathbf{F}_h(Cx_\tau) = Y(Cx_\tau)M^{-1}\), and
with the constant $\eta=\mathsf{Tr}(\mathcal D^\top Q^{-1}\mathcal D),$ which completes the proof.
\end{proof}

Algorithm~\ref{Alg2} summarizes the overall procedure for synthesizing the CBC and its associated PSC for stochastic networked control systems subject to uplink delays and downlink packet losses, and for computing the resulting finite-horizon probabilistic safety guarantee.

\section{Case Studies}\label{sec: Case}
We demonstrate the effectiveness of the proposed framework through three case studies: an academic system, a jet engine compressor~\cite{anta2010sample}, and a Van der Pol oscillator~\cite{Han_3D_VAN}. The primary objective is to synthesize a CBC and its associated PSC for a dt-A-PSNCS in \eqref{eq:theta-affine}, considering the effects of uplink delays and downlink packet dropouts. This is achieved by following the procedure detailed in Algorithm~\ref{Alg2}. The key quantitative results obtained for all case studies are summarized in Table~\ref{tab:system-configurations}. Where applicable, for notational simplicity, we denote the elements of the packetized controller by $\nu:=\nu_k(k)$, $\nu_1:=\nu_k(k+1)$, $\nu_2:=\nu_k(k+2)$, and $\nu_3:=\nu_k(k+3)$, and the delayed state components by $x_{\tau1}:=x_1(k-\tau)$, $x_{\tau2}:=x_2(k-\tau)$ and $x_{\tau3}:=x_3(k-\tau)$. All simulations were conducted on a \textsf{MacBook} with an \textsf{M2} chip and 32~\textsf{GB} of memory, using \textsf{MATLAB R2023b}.

\subsection{Case Study 1: Academic System}\label{Case_study_1}
We consider a polynomial stochastic plant with the dynamics described as
\begin{align}\notag
	{x}_1(k \!+\!1)&\!=\! x_1(k) \!+\! 0.01x_1(k)x_2(k) \! +\! {u}(k) \! +\! 0.1w_1(k),\\\notag
	{x}_2(k \!+\! 1)& \!=\! x_2(k) \!+\! 0.02x_1(k)x_2(k) \!+\! x_1(k){u}(k) \!+\! 0.1w_2(k),\\\label{academic}
	y(k) &= [x_1(k);x_2(k)].
\end{align}
The system in \eqref{academic} can be expressed in the form of the plant in Definition~\ref{Plant_model} along with its relevant matrices as $	C = \mathbf{I}_2,$ $D = 0.1\times \mathbf{I}_2$ and 
\begin{align*}
	{A}(x) \!=\! \begin{bmatrix}
		1 & 0.01x_1\\
		0.02x_2 & 1
	\end{bmatrix}\!\!,\quad
	G(x) \!=\!	\begin{bmatrix}
		1 \\
		x_1 
	\end{bmatrix}\!\!.
\end{align*}
The time-invariant uplink delay and the downlink packet dropout rate are given as $\tau=2$ and $p=0.03$, respectively. 

The regions of interest are given by $\mathbb X= [-4,4]^2$, $\mathbb X_{a} = [-1,1]^2$, and $\mathbb X_{b} = [3, 4]^2 \cup [-4,-3]^2$. Following the steps outlined in Algorithm~\ref{Alg2}, given an allowable outage probability of $\alpha = 0.0127$, and the fixed time
horizon $\mathcal{T}=200$, Lemma~\ref{lem:run_bernoulli} and the relation in \eqref{eq:N_design} yield a control input prediction horizon of $h = 3$. Then, through the satisfaction of condition \eqref{eq:SOS-inverse-condition}, we compute the CBC matrix\footnote{Due to the large dimension of $P \in \mathbb{R}^{9 \times 9}$, we omit its full description for brevity.} $P$ and its corresponding PSC as
\begin{equation}\label{cont_A}
	\boldsymbol{u} \!=\! \begin{bmatrix}
		-0.015x_{\tau 2}^2-0.018x_{\tau 1}\\
		-0.001x_{\tau 1}\\
		-0.00018x_{\tau 1}x_{\tau 2} -0.0008x_{\tau 1}
	\end{bmatrix}
	\!\!.
\end{equation}
Given the synthesized matrix \(P\) and the conditions in~\eqref{eq:gamma-a-inverse} and~\eqref{eq:gamma-b-inverse}, we design the parameters
\(\gamma_a =  6.06 \times 10^{6}\), \(\gamma_b = 9.34 \times 10^{7}\) and $\eta=2.05 \times 10^{3}$ in \eqref{eq:eta-inverse}.
Therefore, by Theorem~\ref{Th: safety}, the system in \eqref{academic} is guaranteed to be safe with a probability of at least $93\%$ over a finite time horizon $\mathcal{T}=200$, even in the presence of uplink time-invariant delay and downlink packet dropouts. Simulation results for the system in \eqref{academic} are shown in Fig.~\ref{fig:AAtraj}. 

{We now re-examine the academic system under a more challenging communication
setting, with uplink delay $\tau=3$ and dropout rate $p=0.07$, keeping the
horizon $\mathcal{T}=200$. For $\alpha=0.0127$,~\eqref{eq:N_design} now gives $h=4$, raising the augmented dimension to
$\kappa=n(\tau+1)+mh=12$. Satisfying condition~\eqref{eq:SOS-inverse-condition} yields the CBC
matrix $P\in\mathbb{R}^{12\times 12}$ and the PSC $\boldsymbol{u}=[\nu;\nu_1;\nu_2;\nu_3]$,
where
\begin{align*}
	\nu   &= -0.02\,x_{\tau1} - 0.01\,x_{\tau2}^2, &
	\nu_1 &= -0.0012\,x_{\tau1}, \\
	\nu_2 &= -0.0009\,x_{\tau1}, &
	\nu_3 &= -0.0006\,x_{\tau1}.
\end{align*}
With $\gamma_a=6.66\times10^{6}$, $\gamma_b=1.02\times10^{8}$, and
$\eta=5.86\times10^{3}$, Theorem~\ref{Th: safety} guarantees safety with
probability at least $92\%$ over $\mathcal{T}=200$, despite the increased
delay and packet loss. A higher dropout rate results in larger values of $h$,
while the increased delay and longer input prediction horizon together enlarge the augmented
dimension $\kappa$; both effects increase the computational burden and the
size of the SOS program.}

\begin{figure*}[t!]
	\centering
	\begin{subfigure}[t]{0.23\textwidth}
		\centering
		\includegraphics[width=\linewidth]{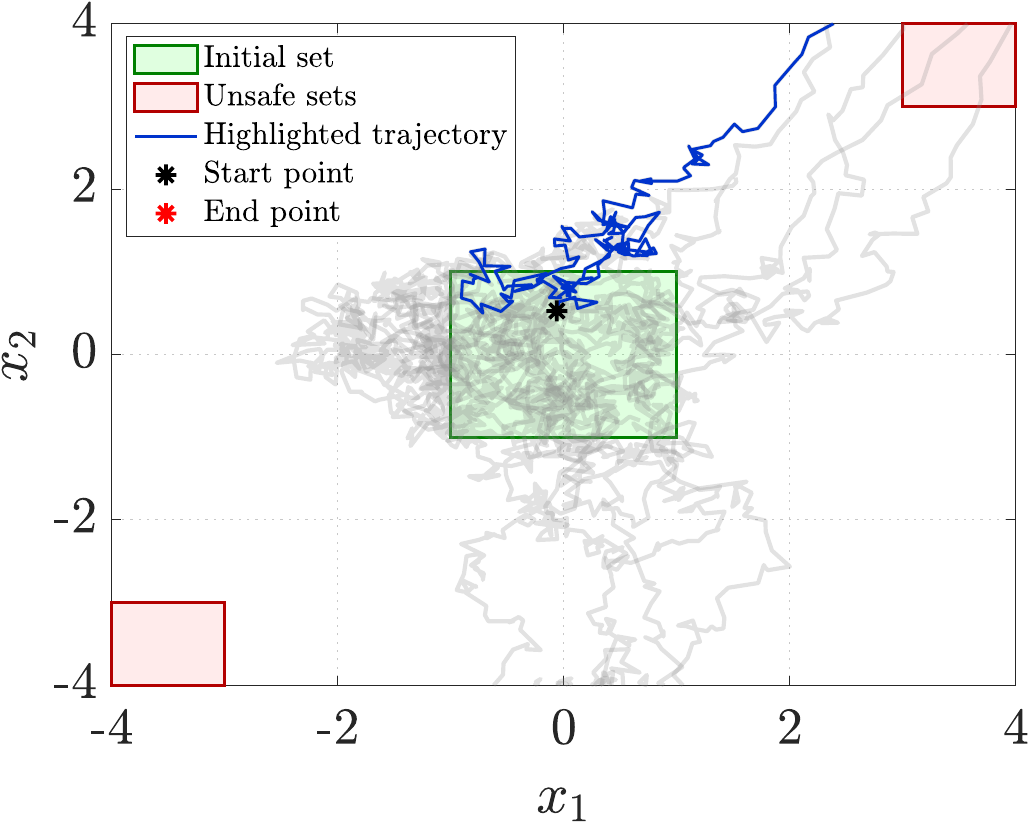}
		\caption{}
		\label{fig:A4}
	\end{subfigure}
	\hfill
	\begin{subfigure}[t]{0.23\textwidth}
		\centering
		\includegraphics[width=\linewidth]{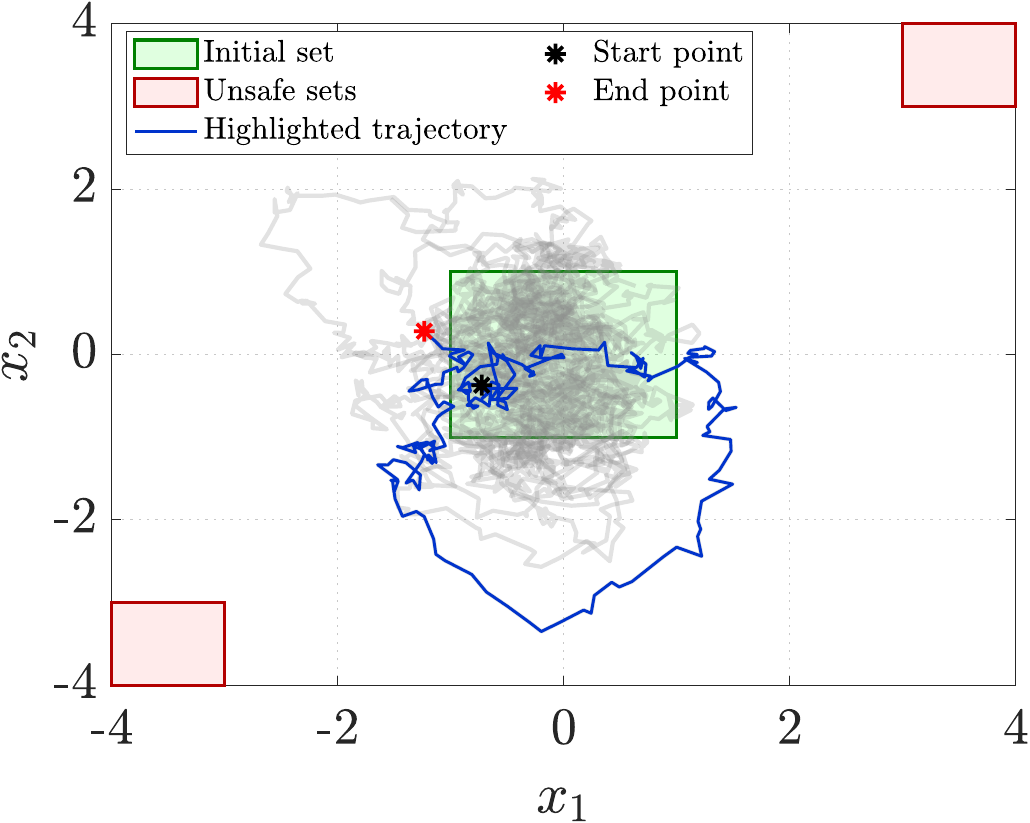}
		\caption{}
		\label{fig:A3}
	\end{subfigure}
	\hfill
	\begin{subfigure}[t]{0.25\textwidth}
		\centering
		\includegraphics[width=\linewidth]{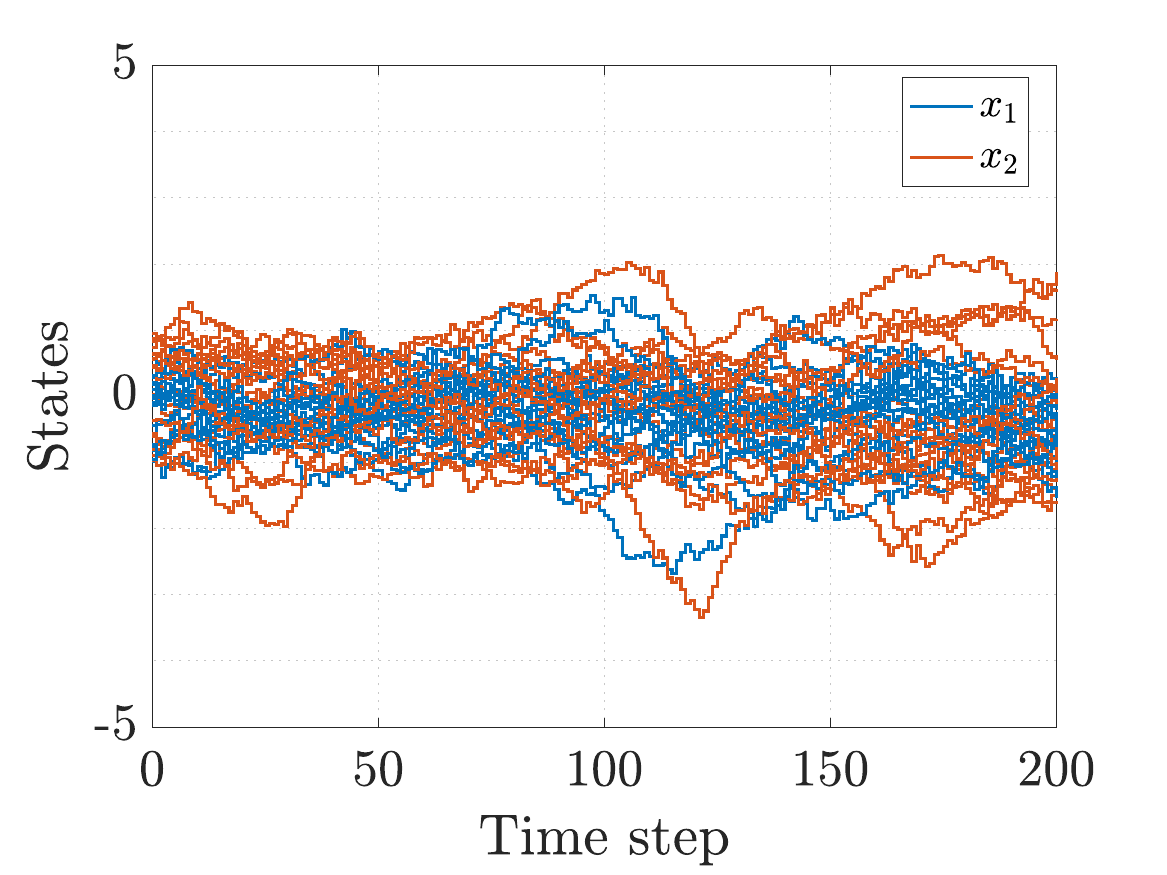}
		\caption{}
		\label{fig:A2}
	\end{subfigure}
	\hfill
	\begin{subfigure}[t]{0.25\textwidth}
		\centering
		\includegraphics[width=\linewidth]{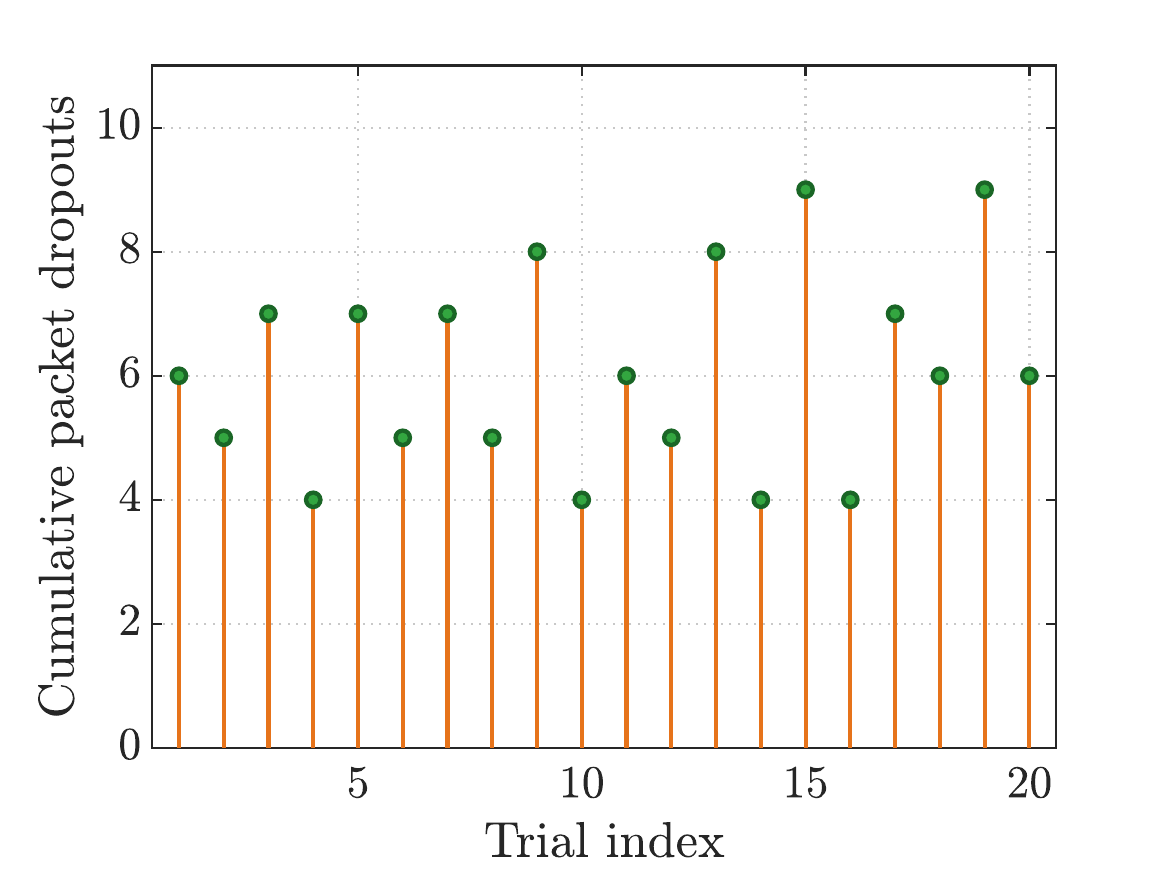}
		\caption{}
		\label{fig:A1}
	\end{subfigure}
	\caption{\textbf{Academic system}. Plot (a) shows the open-loop state trajectories, whereas plot (b) presents the closed-loop trajectories generated by the synthesized PSC in~\eqref{cont_A}, for several initial conditions selected from $\mathbb{X}_a \subset [-1,1]^2$. Both plots are generated using $20$ different arbitrary noise realizations. Plot (c) depicts the state evolution over the safety horizon $\mathcal{T}=200$, confirming satisfaction of the safety specification~$\Upsilon$. Plot (d) illustrates the cumulative number of packet dropouts over the horizon $\mathcal{T}=200$ for each of the $20$ trials corresponding to different noise realizations.}
	\label{fig:AAtraj}
\end{figure*}
\begin{figure*}[t!]
	\centering
	\begin{subfigure}[t]{0.23\textwidth}
		\centering
		\includegraphics[width=\linewidth]{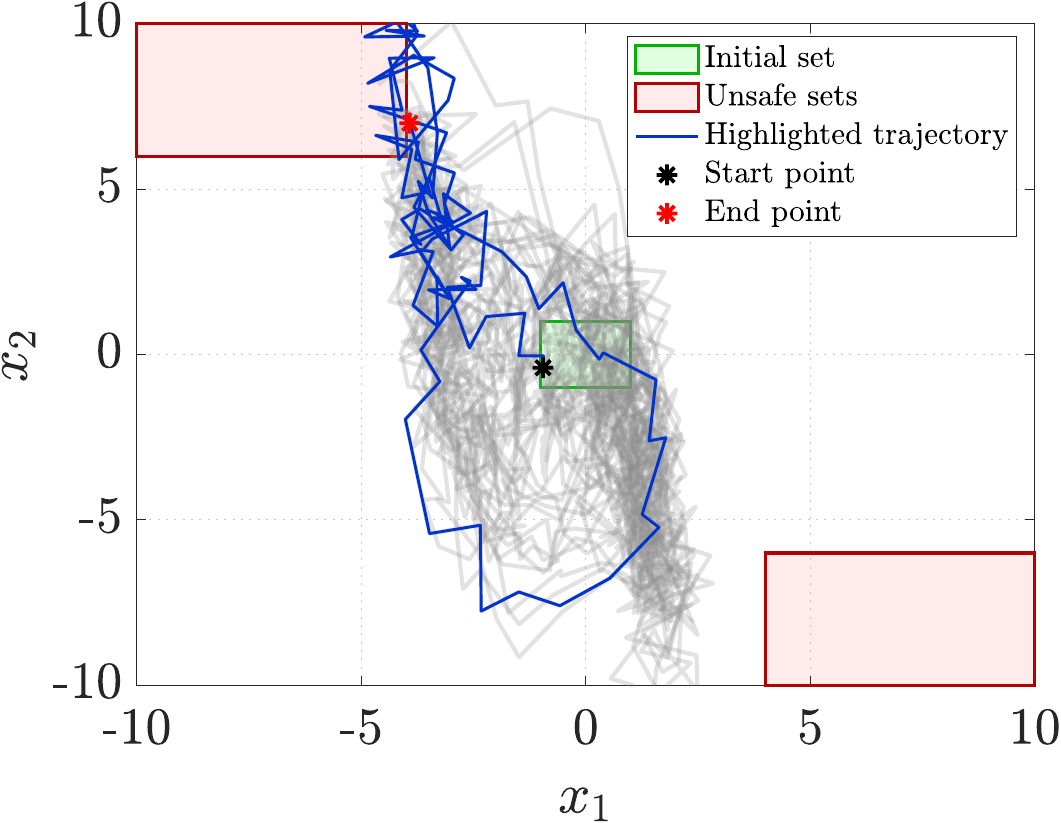}
		\caption{}
		\label{fig:J4}
	\end{subfigure}
	\hfill
	\begin{subfigure}[t]{0.23\textwidth}
		\centering
		\includegraphics[width=\linewidth]{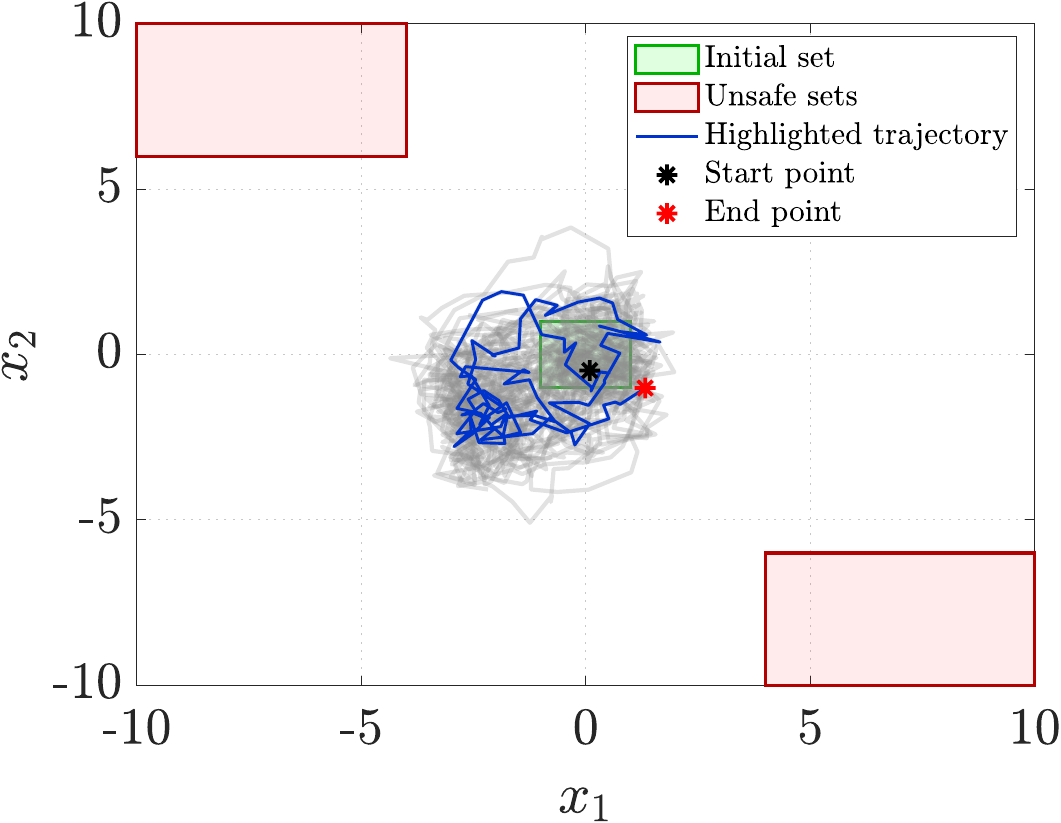}
		\caption{}
		\label{fig:J3}
	\end{subfigure}
	\hfill
	\begin{subfigure}[t]{0.25\textwidth}
		\centering
		\includegraphics[width=\linewidth]{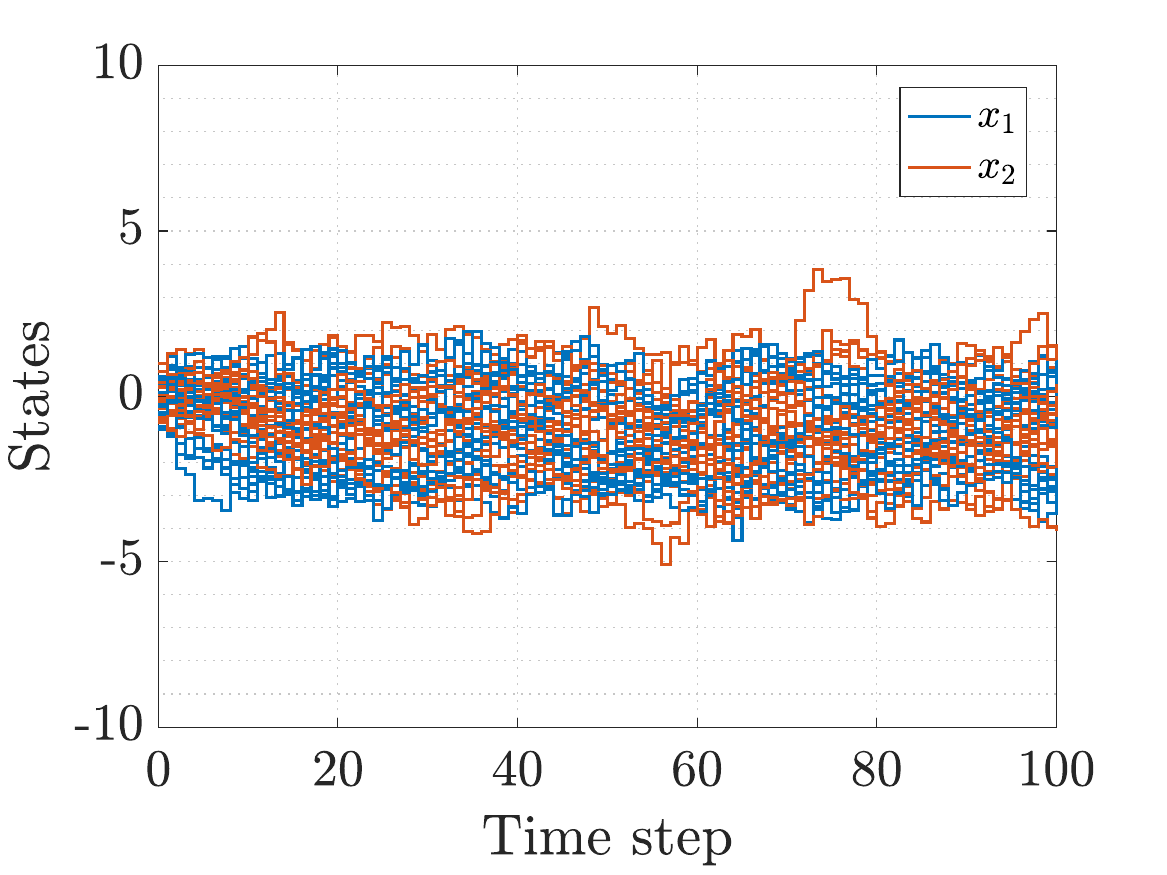}
		\caption{}
		\label{fig:J2}
	\end{subfigure}
	\hfill
	\begin{subfigure}[t]{0.25\textwidth}
		\centering
		\includegraphics[width=\linewidth]{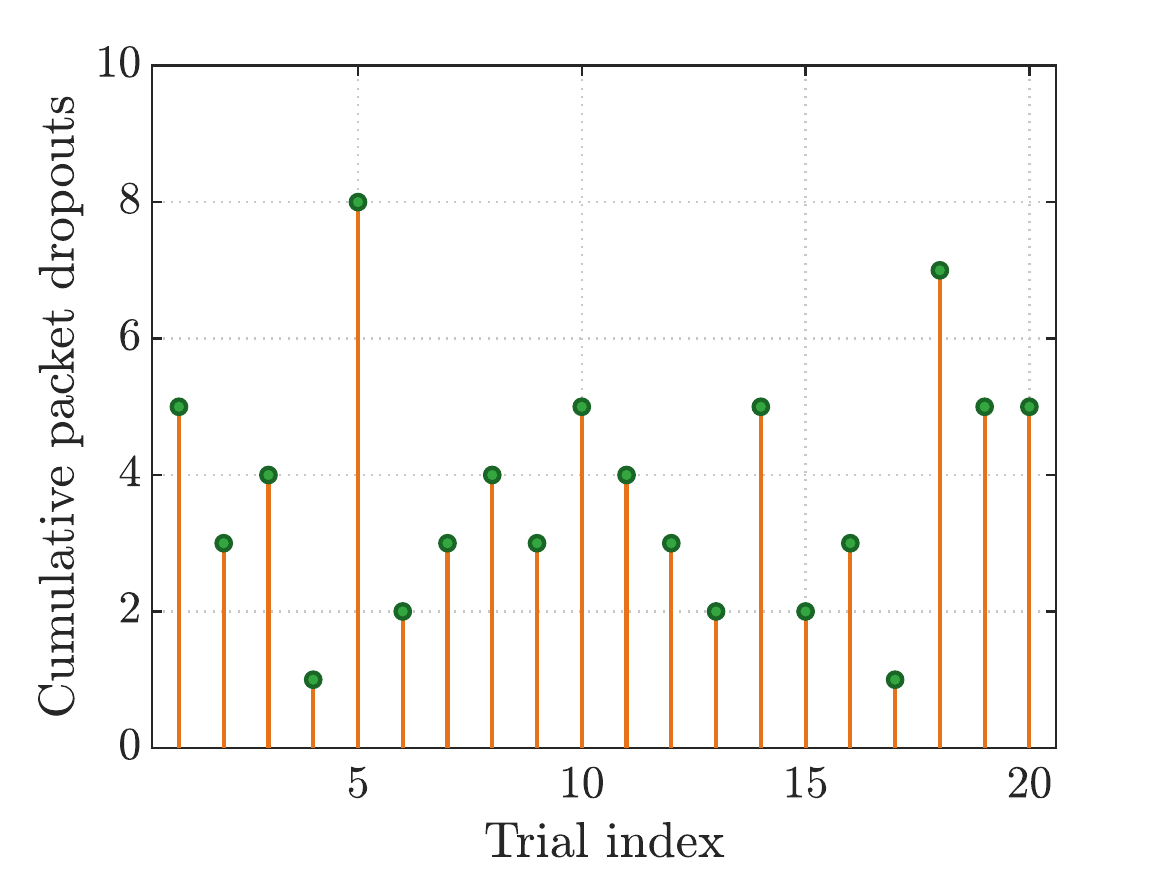}
		\caption{}
		\label{fig:J1}
	\end{subfigure}
	\caption{\textbf{Jet engine compressor}. Plot (a) displays the trajectories under a random controller, while plot (b) displays the trajectories under the designed PSC in~\eqref{cont_jet}, starting from different initial conditions within $\mathbb{X}_a \subset [-1,1]^2$. Both simulations are generated using $20$ distinct arbitrary noise realizations. Plot (c) depicts the trajectories over safety horizon $\mathcal{T}=100$, demonstrating compliance with the specified safety property~$\Upsilon$. Plot (d) shows the cumulative packet dropouts over the safety horizon $\mathcal{T}=100$ for each of the $20$ trials, corresponding to different arbitrary noise realizations.}
	\label{fig:Atraj}
\end{figure*}

\subsection{Case Study 2: Jet Engine Compressor}\label{Case_study_2}
\vspace{-1mm}
In the second case study, we examine a physical polynomial stochastic system, a jet engine compressor~\cite{anta2010sample}, with dynamics given by
\begin{align}\notag
	{x}_1(k+1)& = x_1(k) - 0.1x_2(k) - 0.15 x^2_1(k)- 0.05 x^3_1(k)\\\notag &~~~ +0.4w_1(k),\\\notag
	{x}_2(k+1) & \!=\! x_2(k) \!+\! 0.1x_1(k)  \!+\! 0.1u(k) \!+\! 0.4w_2(k),\\\label{jet}
	y(k) &= [x_1(k);x_2(k)],
\end{align}
where $x_1$ denotes the mass flow, $x_2$ is the pressure rise, and $u$ represents the control input. The system described in \eqref{jet} can be represented in the plant form as outlined in Definition~\ref{Plant_model}, along with its matrices $	C=  \mathbf{I}_2,$ $D = 0.4 \times \mathbf{I}_2,$ and
\begin{align*}
	{A}(x) \!=\! \begin{bmatrix}
		1-0.15x_1-0.05x_1^2 & -0.1\\
		0.1 & 1
	\end{bmatrix}\!\!,\quad
	G \!=\!	\begin{bmatrix}
		0 \\
		0.1
	\end{bmatrix}\!\!.
\end{align*}
Unlike the first case study, where the maximum degree of the system monomials was $2$, the jet engine compressor features a maximum monomial degree of $3$. The time-invariant uplink delay and the downlink packet dropout rate are given as $\tau=2$ and $p=0.04$, respectively. 

The regions of interest are defined as $\mathbb{X} = [-10,10]^2$, $\mathbb{X}_a = [-1,1]^2$, and $\mathbb{X}_b = [4,10]\times[-10,-6] \cup [-10,-4] \times [6,10]$. Given an allowable outage probability of $\alpha = 0.054$ and the time horizon $\mathcal{T}=100$, \eqref{eq:N_design} provides a control input prediction horizon of $h = 3$. Subsequently, by satisfying condition \eqref{eq:SOS-inverse-condition}, we compute the CBC matrix $P \in \mathbb{R}^{9 \times 9}$ and derive its corresponding PSC as
\begin{equation}\label{cont_jet}
	\boldsymbol{u} \!=\! [\nu;\nu_1;\nu_2],
\end{equation}
where
\begin{align*}
	\nu =& 0.003x_{\tau1}^2x_{\tau2}+0.002x_{\tau2}^3-0.03x_{\tau1}^2 \nonumber\\
	&-0.02x_{\tau2}^2-0.2x_{\tau1}-x_{\tau2}, \nonumber\\[0.3em]
	\nu_1 =& 0.006x_{\tau1}^2x_{\tau2}+0.005x_{\tau2}^3-0.03x_{\tau1}^2 \nonumber\\
	&-0.02x_{\tau2}^2-0.1x_{\tau1}-0.49x_{\tau2},\\
	\nu_2 =& 0.009x_{\tau1}^2x_{\tau2}+0.007x_{\tau2}^3-0.03x_{\tau1}^2 \nonumber\\
	&-0.02x_{\tau2}^2-0.06x_{\tau1}-0.33x_{\tau2}.
\end{align*}
We subsequently design
\(\gamma_a =  4.06 \times 10^{5}\), \(\gamma_b = 1.44 \times 10^{7}\), and $\eta= 9.39 \times 10^{3}$. Therefore, by Theorem~\ref{Th: safety}, the system in \eqref{jet} is guaranteed to be safe with a probability of at least $90\%$ over a finite time horizon $\mathcal{T}=100$, in the presence of uplink time-invariant delay and downlink packet dropouts. Simulation results for the jet engine system in \eqref{jet} are illustrated in Fig.~\ref{fig:Atraj}.

\begin{figure*}[t!]
	\centering
	\begin{subfigure}[t]{0.23\textwidth}
		\centering
		\includegraphics[width=\linewidth]{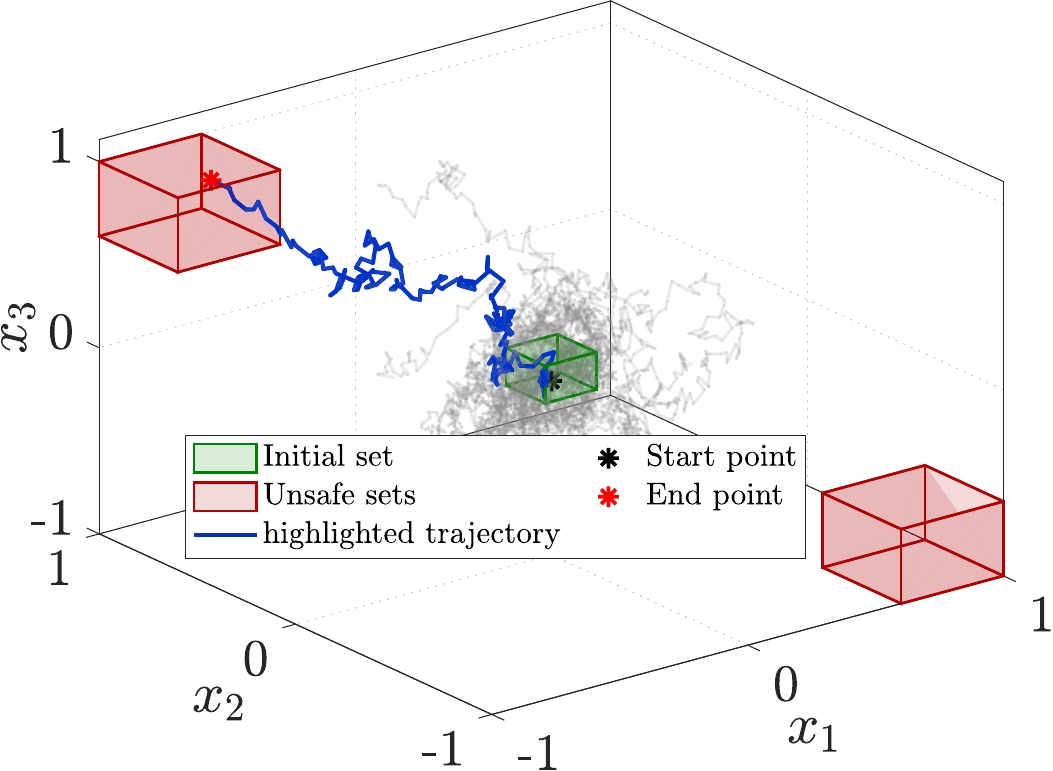}
		\caption{}
		\label{fig:V4}
	\end{subfigure}
	\hfill
	\begin{subfigure}[t]{0.23\textwidth}
		\centering
		\includegraphics[width=\linewidth]{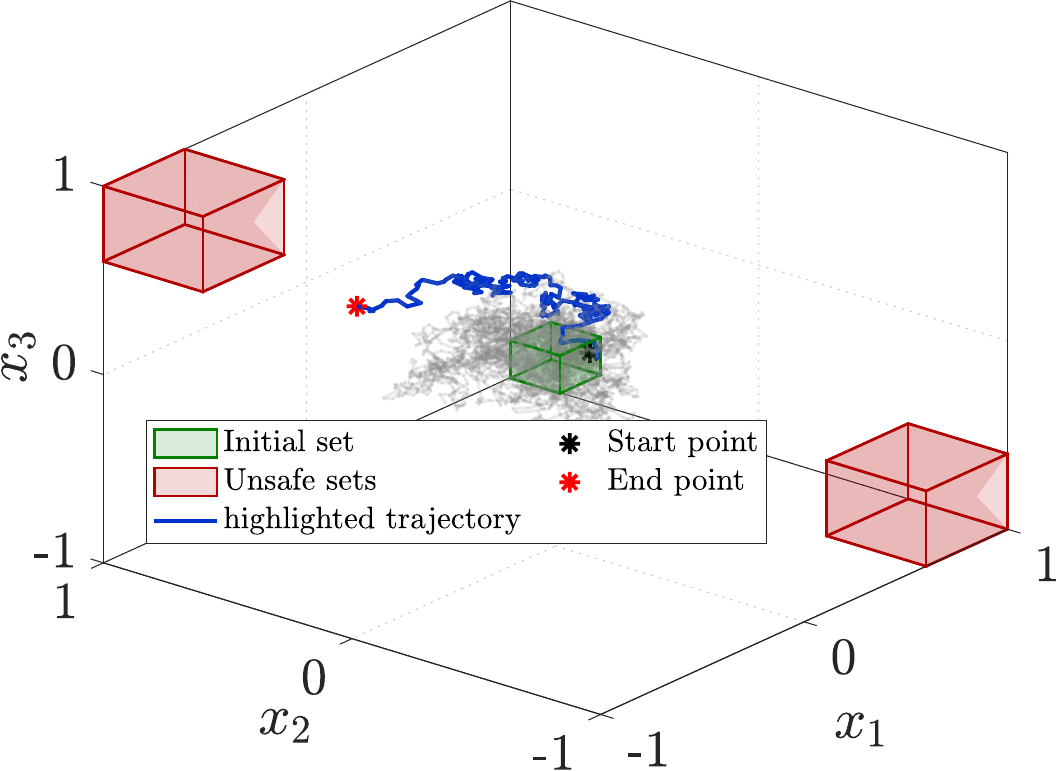}
		\caption{}
		\label{fig:V3}
	\end{subfigure}
	\hfill
	\begin{subfigure}[t]{0.25\textwidth}
		\centering
		\includegraphics[width=\linewidth]{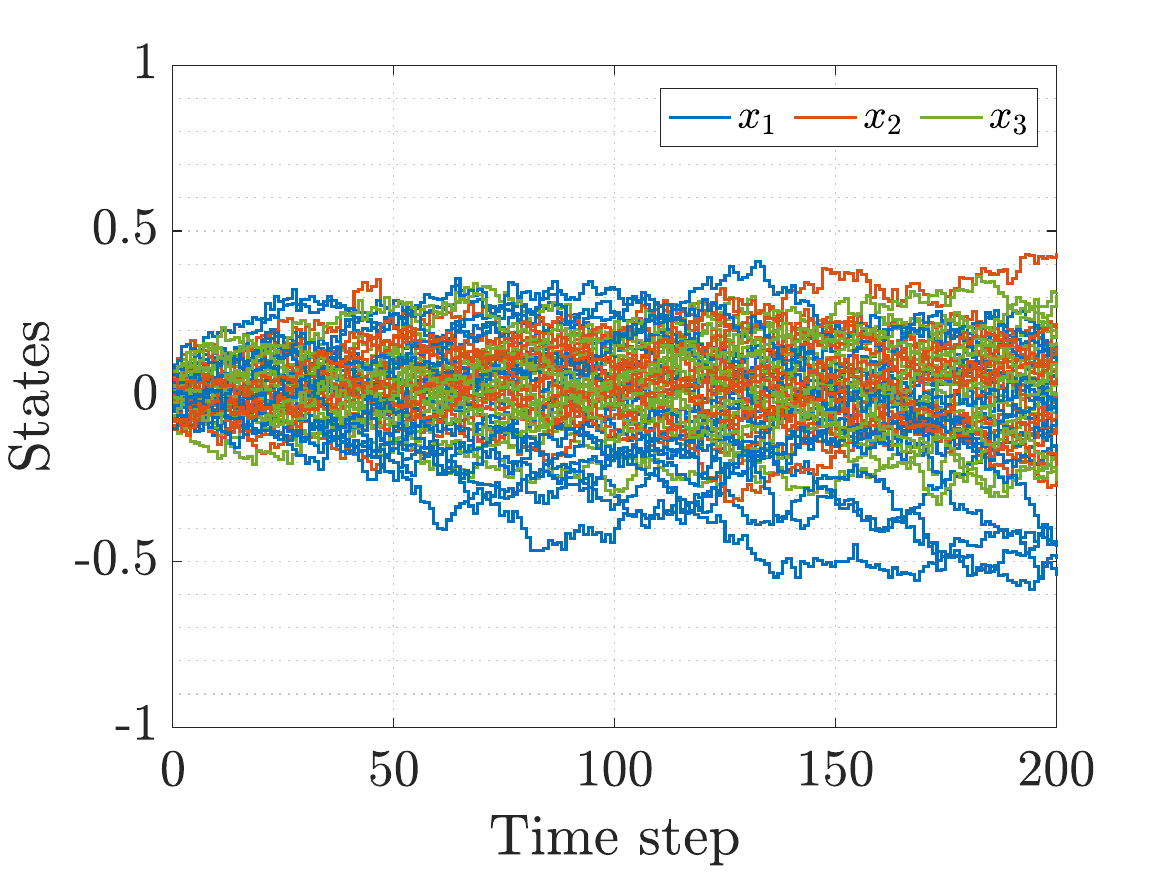}
		\caption{}
		\label{fig:V2}
	\end{subfigure}
	\hfill
	\begin{subfigure}[t]{0.25\textwidth}
		\centering
		\includegraphics[width=\linewidth]{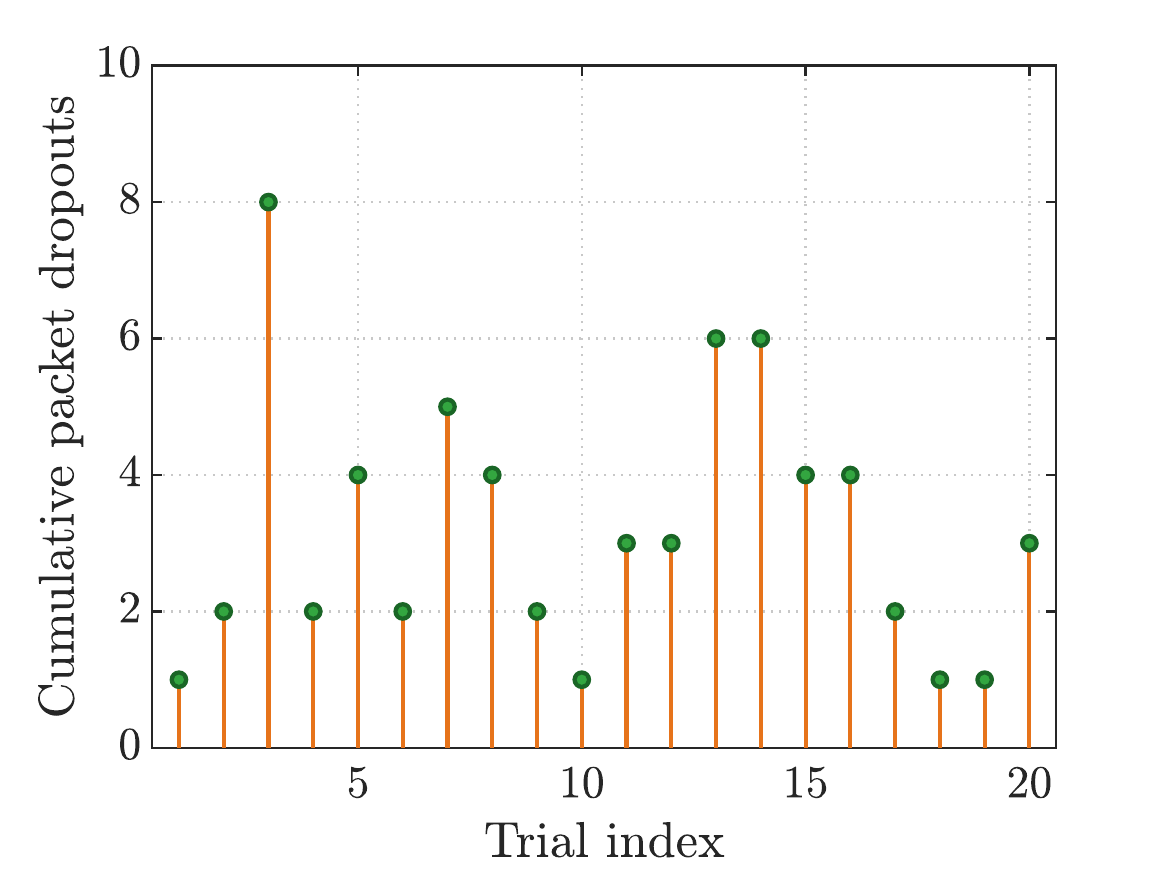}
		\caption{}
		\label{fig:V1}
	\end{subfigure}
	\caption{\textbf{Van der Pol oscillator}. Plot~(a) shows the trajectories under a random controller, while plot~(b) presents the trajectories obtained under the designed PSC in~\eqref{cont_V}, for different initial conditions in $\mathbb{X}_a \subset [-0.1,0.1]^3$. In both cases, the simulations are carried out using $20$ different arbitrary noise realizations. Plot~(c) depicts the system trajectories over the safety horizon $\mathcal{T}=200$, confirming satisfaction of the safety specification~$\Upsilon$. Plot~(d) displays, for each of the $20$ trials, the cumulative number of packet dropouts over the same safety horizon, where each trial corresponds to a distinct arbitrary noise realization.}
	\label{fig:Vtraj}
\end{figure*}
\begin{figure*}[t]
	\centering
	\subfloat[Academic system\label{fig:MC1}]{
		\includegraphics[width=0.31\textwidth]{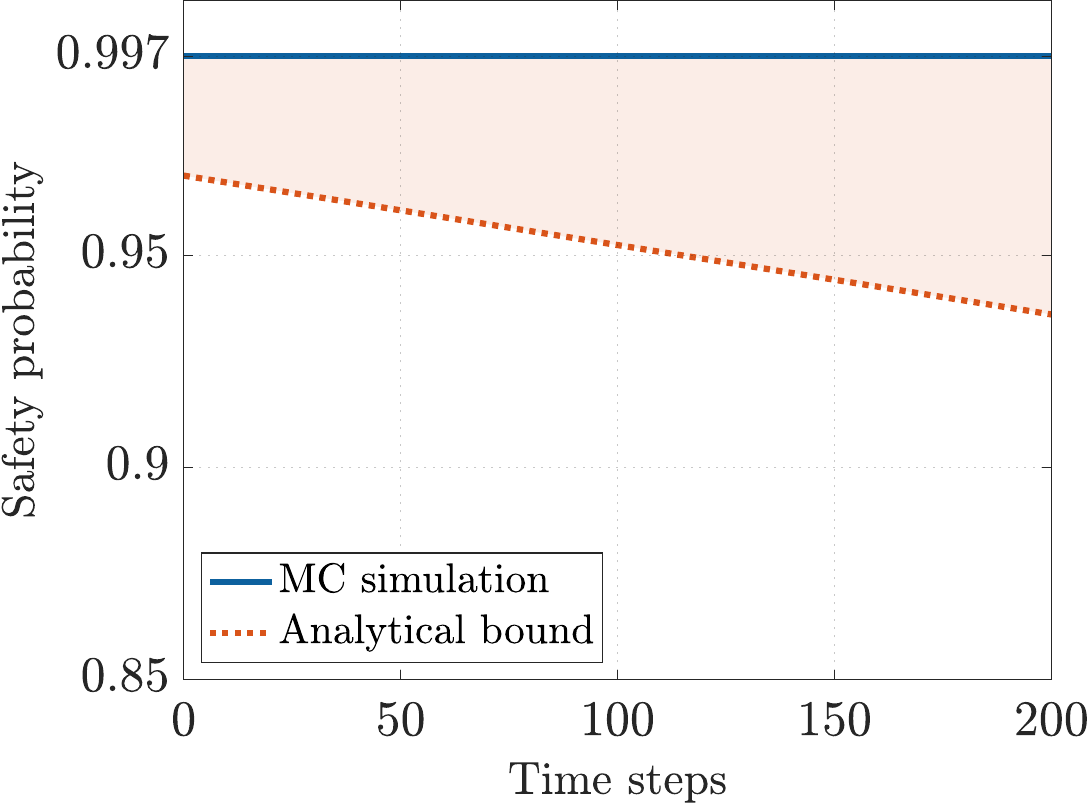}
	}
	\subfloat[Jet engine compressor\label{fig:MC2}]{
		\includegraphics[width=0.31\textwidth]{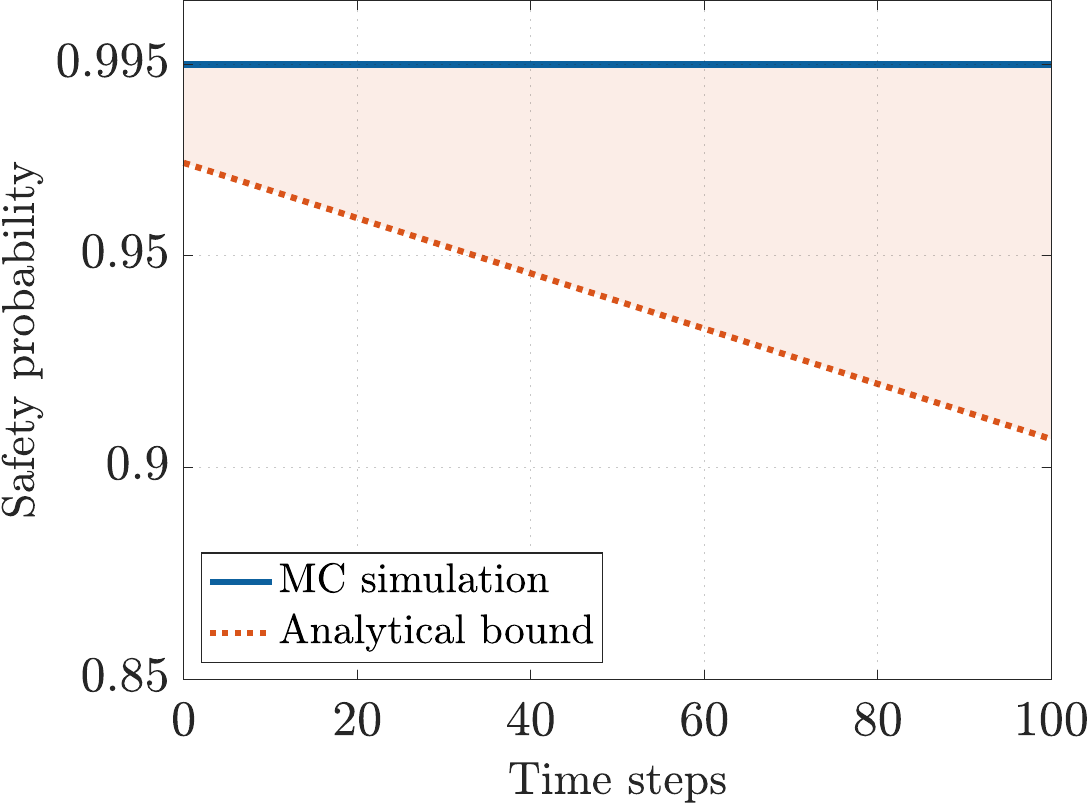}
	}
	\subfloat[Van der Pol oscillator\label{fig:MC3}]{
		\includegraphics[width=0.31\textwidth]{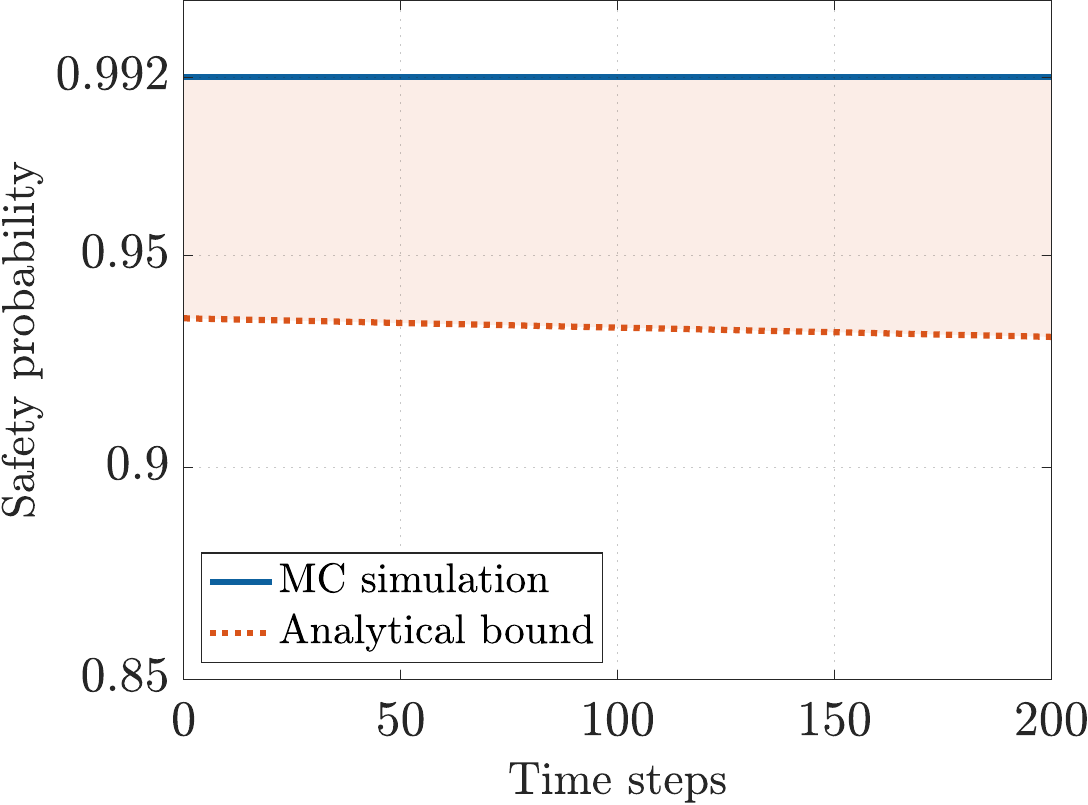}
	}
	\caption{Comparison between the analytical safety lower bounds in \eqref{bound} and MC safety probabilities for the obtained PSC. Solid curves show the empirical MC safety probabilities obtained from sampled closed-loop trajectories, while dashed curves represent the analytical safety lower bounds certified from the CBC conditions over the finite time horizons.}
	\label{fig:mc_safety_prob}
\end{figure*}

\subsection{Case Study 3: Van der Pol Oscillator}\label{Case_study_3}
As the last case study, we consider a three-dimensional Van der Pol oscillator \cite{Han_3D_VAN}, whose dynamics are given by
\begin{align} \notag
	x_1(k \!+\! 1) &= x_1(k)-0.02\,x_2(k) +0.02w_1(k), \notag\\
	x_2(k \!+\! 1) &= x_2(k)+0.008\,x_1(k)-0.021\,x_2(k)\notag\\ &~~~+0.01\,x_3(k)+0.1\,x_1^2(k)x_2(k) +0.02w_2(k), \notag\\
	x_3(k \!+\! 1) &= 0.99\,x_3(k) \!+\! 0.01\,x_3^3(k) \!+\! 0.005\,u(k) \!+\! 0.02w_3(k), \notag\\
	y(k) &= x_1(k) + x_2(k) + x_3(k).
	\label{vdp3d}
\end{align}
The system in~\eqref{vdp3d} can be expressed in the form of the plant in Definition~\ref{Plant_model}, with the relevant matrices given as $D= 0.02 \times \mathbf{I}_3,$ $C=[
1 \; 1 \;1 ],$ and
\begin{align*}
	A(x) \!\!=\!\!
	\begin{bmatrix}
		1 & -0.02 & 0\\
		0.008 & 0.979+0.1x_1^2 & 0.01\\
		0 & 0 & 0.99+0.01x_3^2
	\end{bmatrix}\!\!,
	G \!\!=\!\!
	\begin{bmatrix}
		0\\
		0\\
		0.005
	\end{bmatrix}\!\!.
\end{align*}
Unlike the previous two case studies, this example is three-dimensional and contains both mixed and cubic polynomial nonlinearities. The time-invariant uplink delay and the downlink packet dropout rate are given as $\tau=2$ and $p=0.015$, respectively. 

The regions of interest are selected as \(\mathbb{X}=[-1,1]^3\), \(\mathbb{X}_a=[-0.1,0.1]^3\), and \(\mathbb{X}_b=[0.6,1] \times [-1,-0.6]^2 \cup [-1,-0.6] \times [0.6,1]^2 \). Given \(\alpha=0.09\) and $\mathcal{T}=200$, we compute \(h = 2\) using~\eqref{eq:N_design}. By enforcing condition~\eqref{eq:SOS-inverse-condition}, the CBC matrix \(P\in \mathbb{R}^{11 \times 11}\) is computed and the corresponding PSC is synthesized as
\begin{align}\label{cont_V}
	\boldsymbol{u} \!=\! [\nu;\nu_1],
\end{align}
where
\begin{align*}
	\nu =& -0.0005x_{\tau1}^2-0.001x_{\tau1}x_{\tau2}-0.001x_{\tau1}x_{\tau3} \nonumber\\
	&-0.0005x_{\tau2}^2-0.001x_{\tau2}x_{\tau3}-0.0005x_{\tau3}^2 \nonumber\\
	&-0.039948(x_{\tau1}+x_{\tau2}+x_{\tau3}),\\
	\nu_1 =& -0.001x_{\tau1}^2-0.002x_{\tau1}x_{\tau2}-0.002x_{\tau1}x_{\tau3} \nonumber\\
	&-0.001x_{\tau2}^2-0.002x_{\tau2}x_{\tau3}-0.001x_{\tau3}^2 \nonumber\\
	&-0.020339(x_{\tau1}+x_{\tau2}+x_{\tau3}).
\end{align*}
Using the obtained matrix  and conditions~\eqref{eq:gamma-a-inverse}--\eqref{eq:gamma-b-inverse}, we design \(\gamma_a=13.34\), \(\gamma_b=427.35\), and \(\eta=0.07\). Hence, by Theorem~\ref{Th: safety}, the system in~\eqref{vdp3d} satisfies the safety specification under the designed controller with a probability of at least $93\%$ for $\mathcal{T}=200$, in the presence of uplink time-invariant delay and downlink packet dropouts. The corresponding simulation results for this case study are presented in Fig.~\ref{fig:Vtraj}.

To assess the conservatism of the analytical safety probability bounds in \eqref{bound}, we compare them with empirical Monte Carlo (MC) safety probabilities for all case studies; see Fig.~\ref{fig:mc_safety_prob}. The number of MC simulations is set to \(10^{5}\), corresponding to independent sampled closed-loop trajectories under the designed PSC. The empirical MC safety probability is at least \(99\%\) in all case studies, exceeding the analytical lower bounds in \eqref{bound}, thereby illustrating the conservative nature of the proposed CBC-based guarantees. Such conservatism is expected, as the analytical bounds are derived from Lyapunov-type sufficient conditions that provide \emph{formal} safety guarantees, whereas MC simulations yield only empirical safety estimates based on a finite number of sampled trajectories. This observation highlights potential opportunities for tightening the theoretical guarantees. In particular, employing higher-degree polynomial CBCs, rather than the quadratic CBCs considered in this work, may lead to less conservative safety conditions and improved analytical safety lower bounds, which is left as future work.

\section{Conclusion}\label{sec: Conclusion}
This work developed a framework for synthesizing safety controllers for discrete-time stochastic networked control systems with time-invariant uplink delays and Bernoulli downlink packet losses, where the underlying plants exhibited polynomial dynamics. To account for these communication imperfections, a plant-side packetized-control buffer and an augmented-state representation were introduced, enabling the enforcement of safety requirements through CBCs and the establishment of quantified probabilistic safety guarantees. The safety constraints were reformulated as SOS conditions and an associated optimization framework, enabling the systematic synthesis of both CBCs and their corresponding PSCs under communication constraints. The proposed approach was validated through three case studies, whose results demonstrated its effectiveness in maintaining safety while satisfying the quantified probabilistic guarantees. Future work will focus on extending the proposed framework to higher-degree polynomial CBCs and to systems with time-varying delays.

\bibliographystyle{ieeetr}
\bibliography{biblio}

\begin{thebibliography}{10}

\bibitem{Hespanha}
J.~P. Hespanha, P.~Naghshtabrizi, and Y.~Xu, ``A survey of recent results in
  networked control systems,'' {\em Proceedings of the IEEE}, vol.~95, no.~1,
  pp.~138--162, 2007.

\bibitem{adimoolam}
A.~Adimoolam and T.~Dang, ``{Safety Verification of Networked Control Systems
  by Complex Zonotopes},'' {\em Leibniz Transactions on Embedded Systems},
  vol.~8, no.~2, pp.~1--22, 2022.

\bibitem{prajna2004safety}
S.~Prajna and A.~Jadbabaie, ``Safety verification of hybrid systems using
  barrier certificates,'' in {\em Proceedings of the International Workshop on
  Hybrid Systems: Computation and Control (HSCC)}, pp.~477--492, 2004.

\bibitem{borrmann2015control}
U.~Borrmann, L.~Wang, A.~D. Ames, and M.~Egerstedt, ``Control barrier
  certificates for safe swarm behavior,'' {\em IFAC-PapersOnLine}, vol.~48,
  no.~27, pp.~68--73, 2015.

\bibitem{ames2019control}
A.~D. Ames, S.~Coogan, M.~Egerstedt, G.~Notomista, K.~Sreenath, and P.~Tabuada,
  ``Control barrier functions: Theory and applications,'' in {\em Proceedings
  of the 18th European Control Conference (ECC)}, pp.~3420--3431, 2019.

\bibitem{santoyo2021barrier}
C.~Santoyo, M.~Dutreix, and S.~Coogan, ``A barrier function approach to
  finite-time stochastic system verification and control,'' {\em Automatica},
  vol.~125, 2021.

\bibitem{nejati2024context}
A.~Nejati, S.~Prakash~Nayak, and A.-K. Schmuck, ``Context-triggered games for
  reactive synthesis over stochastic systems via control barrier
  certificates,'' in {\em Proceedings of the 27th ACM International Conference
  on Hybrid Systems: Computation and Control}, pp.~1--12, 2024.

\bibitem{ZAKER2026113082}
M.~Zaker, O.~Akbarzadeh, B.~Samari, and A.~Lavaei, ``Compositional design of
  safety controllers for large-scale stochastic hybrid systems,'' {\em
  Automatica}, vol.~190, 2026.

\bibitem{lavaei2024scalable}
A.~Lavaei and E.~Frazzoli, ``Scalable synthesis of safety barrier certificates
  for networks of stochastic switched systems,'' {\em IEEE Transactions on
  Automatic Control}, vol.~69, no.~11, pp.~7294--7309, 2024.

\bibitem{lavaei2022automated}
A.~Lavaei, S.~Soudjani, A.~Abate, and M.~Zamani, ``Automated verification and
  synthesis of stochastic hybrid systems: {A} survey,'' {\em Automatica},
  vol.~146, 2022.

\bibitem{Amesdelay2023}
T.~G. Molnar, A.~K. Kiss, A.~D. Ames, and G.~Orosz, ``{S}afety-{C}ritical
  {C}ontrol with {I}nput {D}elay in {D}ynamic {E}nvironment,'' {\em IEEE
  Transactions on Control Systems Technology}, vol.~31, no.~4, pp.~1507--1520,
  2023.

\bibitem{REN2022}
W.~Ren, R.~M. Jungers, and D.~V. Dimarogonas, ``Razumikhin and {K}rasovskii
  {A}pproaches for {S}afe {S}tabilization,'' {\em Automatica}, vol.~146, 2022.

\bibitem{akbarzadeh2026safety}
O.~Akbarzadeh, M.~H. Ashoori, A.~Nejati, and A.~Lavaei, ``Safety controller
  synthesis for stochastic polynomial time-delayed systems,'' {\em arXiv:
  2602.06569}, 2026.

\bibitem{Akbarzadeh_Packet}
O.~Akbarzadeh, S.~Soudjani, and A.~Lavaei, ``Safety barrier certificates for
  stochastic control systems with wireless communication networks,'' in {\em
  IEEE 63rd Conference on Decision and Control (CDC)}, pp.~5185--5190, 2024.

\bibitem{Akbarzadeh_CoDIT}
O.~Akbarzadeh, A.~Nejati, and A.~Lavaei, ``Data-driven safety controller
  synthesis for unknown systems with wireless communication networks,'' in {\em
  10th International Conference on Control, Decision and Information
  Technologies (CoDIT)}, pp.~2229--2234, 2024.

\bibitem{4118454}
J.~Baillieul and P.~J. Antsaklis, ``Control and {C}ommunication {C}hallenges in
  {N}etworked {R}eal-time {S}ystems,'' {\em Proceedings of the IEEE}, vol.~95,
  no.~1, pp.~9--28, 2007.

\bibitem{4118476}
L.~Schenato, B.~Sinopoli, M.~Franceschetti, K.~Poolla, and S.~S. Sastry,
  ``{F}oundations of {C}ontrol and {E}stimation {O}ver {L}ossy {N}etworks,''
  {\em Proceedings of the IEEE}, vol.~95, no.~1, pp.~163--187, 2007.

\bibitem{ip-cta_20050178}
T.~Yang, ``Networked {C}ontrol {S}ystem: {A} {B}rief {S}urvey,'' {\em IEE
  Proceedings - Control Theory and Applications}, vol.~153, pp.~403--412(9),
  2006.

\bibitem{5779066}
F.~Ferrari, M.~Zimmerling, L.~Thiele, and O.~Saukh, ``Efficient {N}etwork
  {F}looding and {T}ime {S}ynchronization with {G}lossy,'' in {\em ACM/IEEE
  International Conference on Information Processing in Sensor Networks},
  pp.~73--84, 2011.

\bibitem{6004816}
Y.~Mo and B.~Sinopoli, ``Kalman {F}iltering with {I}ntermittent {O}bservations:
  {T}ail {D}istribution and {C}ritical {V}alue,'' {\em IEEE Transactions on
  Automatic Control}, vol.~57, no.~3, pp.~677--689, 2012.

\bibitem{8405590}
D.~Maity, M.~H. Mamduhi, S.~Hirche, K.~H. Johansson, and J.~S. Baras, ``Optimal
  {LQG} {C}ontrol {U}nder {D}elay-dependent {C}ostly {I}nformation,'' {\em IEEE
  Control Systems Letters}, vol.~3, no.~1, pp.~102--107, 2019.

\bibitem{yu2005stabilization}
M.~Yu, L.~Wang, T.~Chu, and F.~Hao, ``Stabilization of networked control
  systems with data packet dropout and transmission delays: Continuous-time
  case,'' {\em European Journal of Control}, vol.~11, no.~1, pp.~40--49, 2005.

\bibitem{9462479}
D.~Maity, M.~H. Mamduhi, S.~Hirche, and K.~H. Johansson, ``Optimal {LQG}
  {C}ontrol of {N}etworked {S}ystems {U}nder {T}raffic-correlated {D}elay and
  {D}ropout,'' {\em IEEE Control Sys. Letters}, vol.~6, pp.~1280--1285, 2022.

\bibitem{maggio_et_al}
M.~Maggio, A.~Hamann, E.~Mayer-John, and D.~Ziegenbein, ``Control-{S}ystem
  {S}tability {U}nder {C}onsecutive {D}eadline {M}isses {C}onstraints,'' in
  {\em Euromicro Conference on Real-Time Systems (ECRTS)}, vol.~165,
  pp.~21:1--21:24, 2020.

\bibitem{Trimpe-stability1}
F.~Mager, D.~Baumann, R.~Jacob, L.~Thiele, S.~Trimpe, and M.~Zimmerling,
  ``Feedback {C}ontrol {G}oes {W}ireless: {G}uaranteed {S}tability {O}ver
  {L}ow-power {M}ulti-hop {N}etworks,'' in {\em ACM/IEEE International
  Conference on Cyber-Physical Systems}, p.~97–108, 2019.

\bibitem{Trimpe-stability2}
F.~Mager, D.~Baumann, C.~Herrmann, S.~Trimpe, and M.~Zimmerling, ``Scaling
  {B}eyond {B}andwidth {L}imitations: Wireless {C}ontrol with {S}tability
  {G}uarantees {U}nder {O}verload,'' {\em ACM Transactions on Cyber-Physical
  Systems (TCPS)}, vol.~6, no.~3, pp.~1--30, 2022.

\bibitem{mamduhi2021delay}
M.~H. Mamduhi, D.~Maity, S.~Hirche, J.~S. Baras, and K.~H. Johansson,
  ``Delay-sensitive joint optimal control and resource management in multiloop
  networked control systems,'' {\em IEEE Transactions on Control of Network
  Systems}, vol.~8, no.~3, pp.~1093--1106, 2021.

\bibitem{mamduhi2025networkaware}
M.~H. Mamduhi and D.~Maity, ``Network-aware optimal sampling for stochastic
  control systems over dynamic networks,'' {\em IEEE Control Systems Letters},
  vol.~9, pp.~1808--1813, 2025.

\bibitem{klugel2019aoi}
M.~Klügel, M.~H. Mamduhi, S.~Hirche, and W.~Kellerer, ``{A}o{I}-penalty
  minimization for networked control systems with packet loss,'' in {\em IEEE
  INFOCOM - IEEE Conference on Computer Communications Workshops (INFOCOM
  WKSHPS)}, pp.~189--196, 2019.

\bibitem{zacchialun2025optimal}
Y.~{Zacchia Lun}, F.~Smarra, and A.~D’Innocenzo, ``Optimal control over
  markovian wireless communication channels under generalized packet dropout
  compensation,'' {\em Automatica}, vol.~176, 2025.

\bibitem{zacchialun2025wireless}
Y.~Zacchia~Lun, F.~Santucci, and A.~D’Innocenzo, ``Wireless control with
  channel state detection and message dropout compensation,'' {\em IEEE Control
  Systems Letters}, vol.~9, pp.~1399--1404, 2025.

\bibitem{jungers2018observability}
R.~M. Jungers, A.~Kundu, and W.~P. M.~H. Heemels, ``Observability and
  controllability analysis of linear systems subject to data losses,'' {\em
  IEEE Transactions on Automatic Control}, vol.~63, no.~10, pp.~3361--3376,
  2018.

\bibitem{dey2014remote}
S.~Dey, A.~Chiuso, and L.~Schenato, ``Remote estimation with noisy measurements
  subject to packet loss and quantization noise,'' {\em IEEE Transactions on
  Control of Network Systems}, vol.~1, no.~3, pp.~204--217, 2014.

\bibitem{leong2008kalman}
A.~S. Leong, S.~Dey, and J.~S. Evans, ``On {K}alman smoothing with random
  packet loss,'' {\em IEEE Transactions on Signal Processing}, vol.~56, no.~7,
  pp.~3346--3351, 2008.

\bibitem{maass2020stochastic}
A.~I. Maass, D.~Ne{\v{s}}i{\'c}, V.~S. Varma, R.~Postoyan, and S.~Lasaulce,
  ``Stochastic stabilisation and power control for nonlinear feedback loops
  communicating over lossy wireless networks,'' in {\em IEEE Conference on
  Decision and Control (CDC)}, pp.~1866--1871, 2020.

\bibitem{BALAGHII201858}
M.~H. {Balaghi I.}, D.~J. Antunes, M.~H. Mamduhi, and S.~Hirche, ``An optimal
  {LQG} controller for stochastic event-triggered scheduling over a lossy
  communication network,'' {\em IFAC-PapersOnLine}, vol.~51, no.~23,
  pp.~58--63, 2018.
\newblock 7th IFAC Workshop on Distributed Estimation and Control in Networked
  Systems NECSYS.

\bibitem{QU2015974}
F.-L. Qu, Z.-H. Guan, D.-X. He, and M.~Chi, ``Event-triggered control for
  networked control systems with quantization and packet losses,'' {\em Journal
  of the Franklin Institute}, vol.~352, no.~3, pp.~974--986, 2015.

\bibitem{mamduhi:MTNS2014}
M.~Mamduhi, A.~Molin, and S.~Hirche, ``Event-based scheduling of multi-loop
  stochastic systems over shared communication channels,'' in {\em 21st
  International Symposium on Mathematical Theory of Networks and Systems
  (MTNS)}, pp.~266--273, 2014.

\bibitem{dhullipalla2025event}
M.~Dhullipalla and D.~V. Dimarogonas, ``Event-triggered predictor-based control
  of networked systems with input delays,'' {\em IEEE Control Systems Letters},
  vol.~9, pp.~2223--2228, 2025.

\bibitem{QuevedoNesic2010NOLCOS}
D.~E. Quevedo and D.~Ne{\v{s}}i{\'c}, ``On stochastic stability of packetized
  predictive control of non-linear systems over erasure channels*,'' {\em IFAC
  Proceedings Volumes}, vol.~43, no.~14, pp.~557--562, 2010.

\bibitem{QuevedoNesic2011ISS}
D.~E. Quevedo and D.~Ne{\v{s}}i{\'c}, ``Input-to-state stability of packetized
  predictive control over unreliable networks affected by packet-dropouts,''
  {\em IEEE Transactions on Automatic Control}, vol.~56, no.~2, pp.~370--375,
  2011.

\bibitem{QuevedoOstergaardNesic2011BitRate}
D.~E. Quevedo, J.~{\O}stergaard, and D.~Nesic, ``Packetized predictive control
  of stochastic systems over bit-rate limited channels with packet loss,'' {\em
  IEEE Transactions on Automatic Control}, vol.~56, no.~12, pp.~2854--2868,
  2011.

\bibitem{Akbarzadeh_Networked}
O.~Akbarzadeh, M.~H. Mamduhi, and A.~Lavaei, ``Safety controller synthesis for
  stochastic networked systems under communication constraints,'' in {\em IEEE
  64th Conference on Decision and Control (CDC)}, pp.~3333--3338, 2025.

\bibitem{BENGTSSON}
F.~Bengtsson and T.~Wik, ``Stochastic optimal control over unreliable
  communication links,'' {\em Automatica}, vol.~142, 2022.

\bibitem{5G-mmWave}
R.~A.~K. Fezeu, E.~Ramadan, W.~Ye, B.~Minneci, J.~Xie, A.~Narayanan, A.~Hassan,
  F.~Qian, Z.~Zhang, J.~Chandrashekar, and M.~Lee, ``An in-depth measurement
  analysis of {5G} mm{W}ave {PHY} latency and its impact on {E}nd-to-{E}nd
  delay,'' pp.~284--312, Springer Nature Switzerland, 2023.

\bibitem{Popovski2019URLLC}
P.~Popovski, {\v{C}}.~Stefanovi{\'c}, J.~J. Nielsen, E.~de~Carvalho,
  M.~Angjelichinoski, K.~F. Trillingsgaard, and A.-S. Bana, ``Wireless access
  in ultra-reliable low-latency communication ({URLLC}),'' {\em IEEE
  Transactions on Communications}, vol.~67, no.~8, pp.~5783--5801, 2019.

\bibitem{Chang2019}
B.~Chang, G.~Zhao, Z.~Chen, L.~Li, and M.~A. Imran, ``Packet-drop design in
  {URLLC} for real-time wireless control systems,'' {\em IEEE Access}, vol.~7,
  pp.~183081--183090, 2019.

\bibitem{Kushner}
H.~J. Kushner, {\em Stochastic Stability and Control}.
\newblock Mathematics in Science and Engineering, Elsevier Science, 1967.

\bibitem{SOSTOOLS}
S.~Prajna, A.~Papachristodoulou, P.~Seiler, and P.~Parrilo, ``{SOSTOOLS}:
  control applications and new developments,'' in {\em Proceeding of IEEE
  International Conference on Robotics and Automation}, pp.~315--320, 2004.

\bibitem{mosek}
M.~ApS, {\em The {MOSEK} {O}ptimization {T}oolbox for {MATLAB} {M}anual.
  Version 10.1.}, 2025.

\bibitem{anta2010sample}
A.~Anta and P.~Tabuada, ``To sample or not to sample: Self-triggered control
  for nonlinear systems,'' {\em IEEE Transactions on automatic control},
  vol.~55, no.~9, pp.~2030--2042, 2010.

\bibitem{Han_3D_VAN}
W.~Han and R.~Tedrake, ``Controller synthesis for discrete-time polynomial
  systems via occupation measures,'' in {\em IEEE/RSJ International Conference
  on Intelligent Robots and Systems}, pp.~6911--6918, 2018.

\end{thebibliography}

\vspace{-1cm}

\begin{IEEEbiography}[{\includegraphics[width=1in,height=1.25in,clip,keepaspectratio]{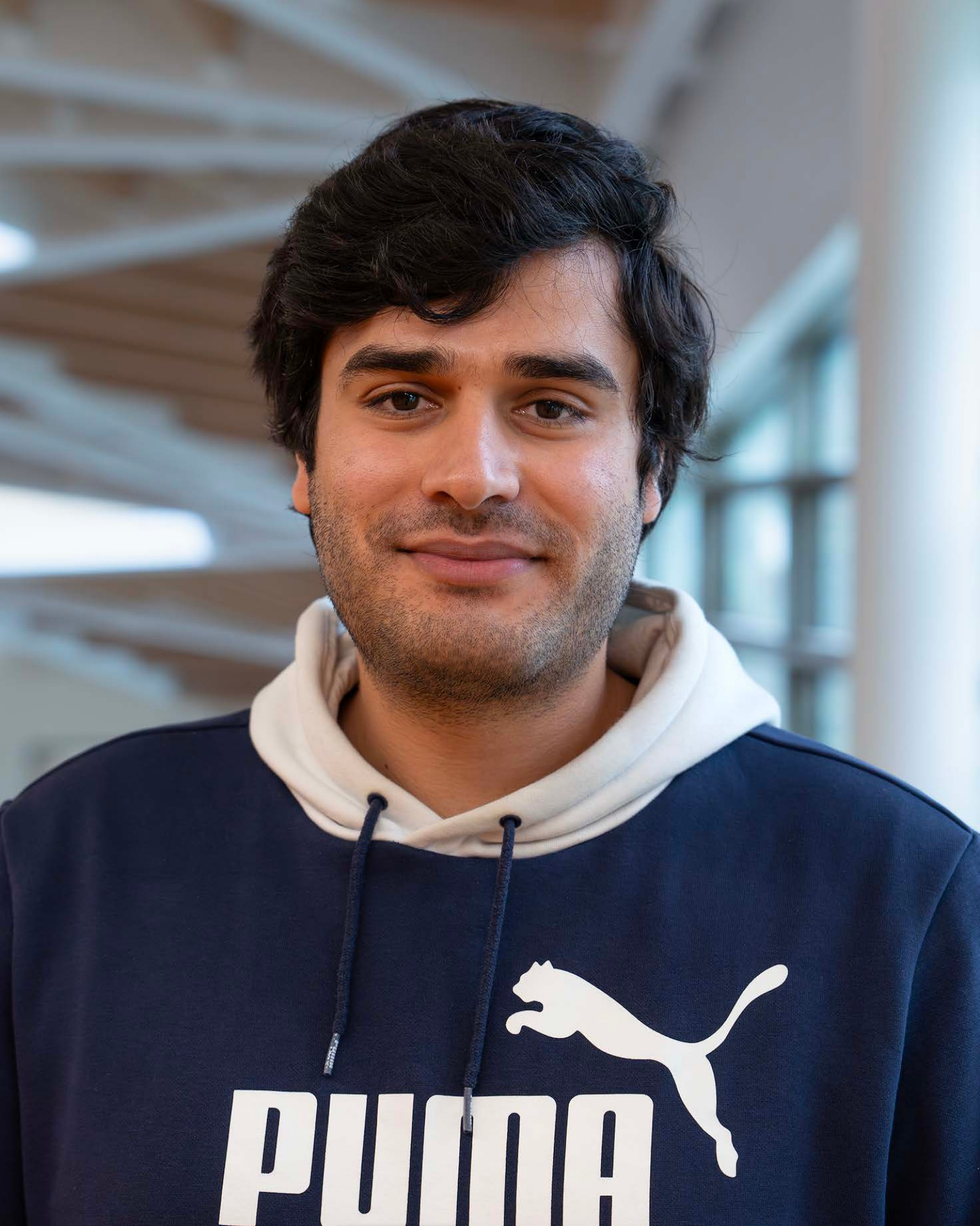}}]{Omid Akbarzadeh}~(Student Member, IEEE) is currently a PhD student in the School of Computing at Newcastle University, U.K. His academic journey commenced at Shiraz University, where he obtained a Bachelor of Science in Electrical and Electronic Engineering. Following this, he pursued a master's degree in Communications and Computer Network Engineering (CCNE) at the Polytechnic University of Turin, Italy (Politecnico di Torino). His research interests include safe cyber-physical systems, communication networks, data-driven approaches, and formal control.
\end{IEEEbiography}\vspace{-1cm}

\begin{IEEEbiography}[{\includegraphics[width=1in,height=1.3in,clip,keepaspectratio]{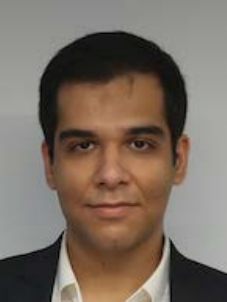}}]{MohammadHossein Ashoori}~(Student Member, IEEE) received his B.Sc. and M.Sc. degrees in Electrical Engineering from Sharif University of Technology (SUT), Tehran, Iran, in 2019 and 2022, respectively. He is currently pursuing his PhD in
	the School of Computing at Newcastle University,
	UK.  His research
	interests include cyber-physical systems (CPS), computer vision, and digital signal processing.
\end{IEEEbiography}\vspace{-1cm}

\begin{IEEEbiography}[{\includegraphics[width=1in,height=1.4in,clip,keepaspectratio]{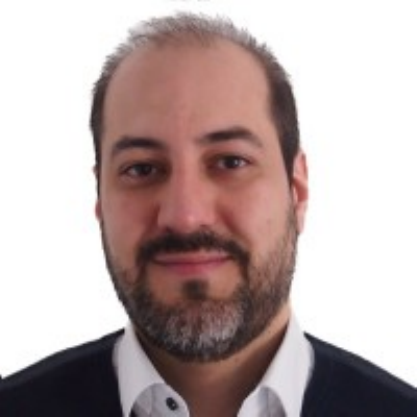}}]{Mohammad H. Mamduhi}~(Senior Member, IEEE) is an Assistant Professor in the School of Computer Science at The University of Birmingham, United Kingdom, the affiliate faculty of the Institute of Data and AI (IDAI), and a Senior IEEE Member. Prior to his current position, he was a Senior Scientist at the Automatic Control Lab, Department of Information Technology and Electrical Engineering, at ETH Zürich, Switzerland. From 2017 to 2020, he was a Postdoctoral and then a Senior Researcher at the Division of Decision and Control Systems, Department of Intelligent Systems, at KTH Royal Institute of Technology, Sweden. He received his PhD from the Department of Electrical and Computer Engineering, Technical University of Munich in Germany (2017), his MSc is Systems, Control, and Robotics from KTH Royal Institute of Technology in Sweden (2010), and his BSc in Mechanical Engineering from Sharif University of Technology (2008). His research interests include networked control systems, safe cyber-physical systems, formal methods, co-design theory, and stochastic evolutionary processes.
\end{IEEEbiography}\vspace{-1cm}

\begin{IEEEbiography}[{\includegraphics[width=1in,height=1.25in,clip,keepaspectratio]{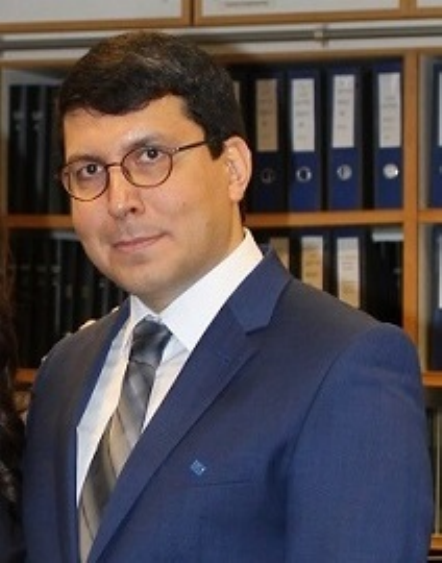}}]{Abolfazl Lavaei}~(M'17--SM'22)
 is an Assistant Professor in the School of Computing at Newcastle University, United Kingdom. Between January 2021 and July 2022, he was a Postdoctoral Associate in the Institute for Dynamic Systems and Control at ETH Zurich, Switzerland. He was also a Postdoctoral Researcher in the Department of Computer Science at LMU Munich, Germany, between November 2019 and January 2021. He received the Ph.D. degree in Electrical Engineering from the Technical University of Munich (TUM), Germany, in 2019. He obtained the M.Sc. degree in Aerospace Engineering with specialization in Flight Dynamics and Control from the University of Tehran (UT), Iran, in 2014. He is the recipient of several international awards in the acknowledgment of his work including Best Repeatability Prize (Finalist) at the ACM HSCC 2025, IFAC ADHS 2024, and IFAC ADHS 2021, HSCC Best Demo/Poster Awards 2022 and 2020, IFAC Young Author Award Finalist 2019, and Best Graduate Student Award 2014 at University of Tehran with the full GPA (20/20). His research interests revolve around the intersection of Control Theory, Formal Methods, and Statistical Learning Theory.
\end{IEEEbiography}

\end{document}